\documentclass{article}
\usepackage[T1]{fontenc}
\usepackage[utf8]{inputenc}
\usepackage[bookmarks=false,hidelinks]{hyperref}
\usepackage{ismir} 
\usepackage{amsmath,amssymb,url,amsthm}
\usepackage{bm}
\usepackage{pifont}
\usepackage{enumitem}
\usepackage{graphicx}
\usepackage{mathtools}
\usepackage{color}
\usepackage[sorting=none,citestyle=numeric-comp,url=false,isbn=false]{biblatex}
\usepackage{lipsum}
\usepackage{booktabs,makecell}
\usepackage{array, multirow, bigdelim} 
\usepackage{tablefootnote}
\usepackage{placeins}
\usepackage[
  activate={true,nocompatibility},
  final,
  tracking=true,
  kerning=true,
  spacing=true,
  protrusion=true,
  expansion=true,
  factor=1100,
  stretch=30,
  shrink=30
]{microtype}
\microtypecontext{spacing=nonfrench}

\usepackage{caption} 
\usepackage{afterpage, longtable}

\newcommand{\cmark}{\ding{51}}%
\newcommand{\xmark}{\ding{55}}%

\newcommand{\calP}{\mathcal{P}}
\newcommand{\calS}{\mathcal{S}}
\newcommand{\calG}{\mathcal{G}}

\newcommand{\bbR}{\mathbb{R}}

\newcommand{\ci}[1]{{\scriptsize $\pm$ #1}}

\newcommand{\ptot}{{\sc Param2Tok}}
\newcommand{\cnfp}{CNF~{(\ptot)}}
\newcommand{\cnfm}{CNF~{(MLP)}}
\newcommand{\cnfe}{CNF~{(Equivariant)}}
\newcommand{\st}{AST}
\newcommand{\vae}{VAE + RealNVP}
\newcommand{\ffnm}{FFN~{(MSE)}}
\newcommand{\ffnc}{FFN~{(Chamfer)}}
\newcommand{\ffns}{FFN~{(Sort)}}

\newcommand{\x}{\mathbf{x}}
\newcommand{\y}{\mathbf{y}}
\newcommand{\Z}{\mathbf{Z}}
\newcommand{\A}{\mathbf{A}}
\newcommand{\X}{\mathbf{X}}
\newcommand{\Y}{\mathbf{Y}}

\newtheorem{theorem}{Theorem}

\newtheorem{definition}{Definition}[section]

\title{Audio synthesizer inversion in symmetric parameter spaces with approximately equivariant flow matching}

\multauthor
  {Ben Hayes \hspace{2cm} Charalampos Saitis \hspace{2cm} Gy{\"o}rgy Fazekas}
  {
  Centre for Digital Music, Queen Mary University of London, United Kingdom \\
  {\tt\small ben@benhayes.net, \{c.saitis, george.fazekas\}@qmul.ac.uk}
  }

\def\authorname{B. Hayes, C. Saitis, G. Fazekas}

\begin{document}

\maketitle

\begin{abstract}
Many audio synthesizers can produce the same signal given different parameter configurations, meaning the inversion from sound to parameters is an inherently ill-posed problem.
We show that this is largely due to intrinsic symmetries of the synthesizer, and focus in particular on permutation invariance.
First, we demonstrate on a synthetic task that regressing point estimates under permutation symmetry degrades performance, even when using a permutation-invariant loss function or symmetry-breaking heuristics.
Then, viewing equivalent solutions as modes of a probability distribution, we show that a conditional generative model substantially improves performance.
Further, acknowledging the invariance of the implicit parameter distribution, we find that performance is further improved by using a permutation equivariant continuous normalizing flow.
To accommodate intricate symmetries in real synthesizers, we also propose a relaxed equivariance strategy that adaptively discovers relevant symmetries from data.
Applying our method to Surge XT, a full-featured open source synthesizer used in real world audio production, we find our method outperforms regression and generative baselines across audio reconstruction metrics.
\end{abstract}

\section{Introduction}\label{sec:introduction}
\begin{figure}[t]
\centerline{\includegraphics[width=\columnwidth]{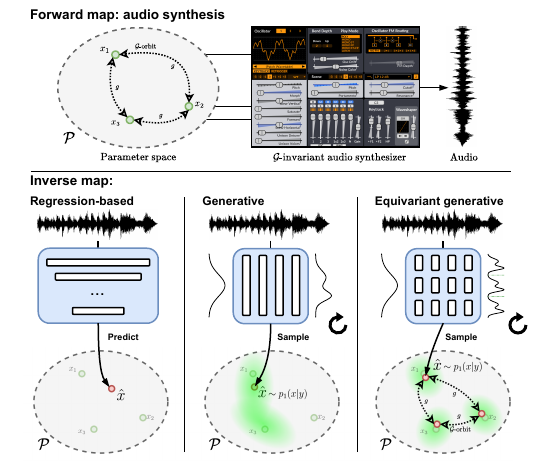}}
  \caption{\textbf{Top: } Audio synthesis is the forward map we wish to invert.
  \textbf{Bottom-left: } Synthesizer inversion by parameter regression. The neural network produces a point estimate, and does not account for symmetries in the synthesizer. \textbf{Bottom-middle: } A generative model approximates the conditional distribution of parameters $\mathbf{x}$ given audio $\mathbf{y}$, but can only learn the appropriate invariances if present in the data. \textbf{Bottom-right: } Using an equivariant flow, the learned distribution is inherently invariant to the symmetries of the synthesizer.}\label{fig:overview}
\vspace{-0.9\baselineskip}
\end{figure}

%

Modern audio synthesizers are intricate systems, combining numerous methods for sound production and manipulation with rich user-facing control schemes.
Where, once, many digital synthesis algorithms were accompanied by a corresponding analysis procedure~\cite{justice_analytic_1979,mcaulay_speech_1986,serra_spectral_1990}, selecting parameters of a modern synthesizer to approximate a given audio signal is a challenging open problem~\cite{roth_comparison_2011,shier_synthesizer_2021} which is increasingly approached using powerful optimization and machine learning algorithms.
In particular, recent works approach the task with deep neural networks trained on datasets sampled directly from the synthesizer~\cite{barkan_inversynth_2019,masuda_improving_2023,bruford_synthesizer_2024}.

Many synthesizers can produce the same signal given multiple different parameter configurations.
This means that inverting the synthesizer is necessarily \textit{ill-posed} --- it lacks a unique solution.
We argue that this harms the performance of models trained to produce point estimates of the parameters.
Maximum likelihood regression objectives that do not account for this ill-posedness are minimized by a suboptimal averaging across equivalent solutions, while invariant loss functions can lead to pathologies such as the \textit{responsibility problem}~\cite{hayes_responsibility_2023,zhang_fspool_2019}.
This, we suggest, explains the superior performance of generative methods when sound matching complex synthesizers~\cite{esling_flow_2020,vaillant_improving_2021} --- for a given input, they can assign predictive weight to many possible parameter configurations, as opposed to selecting just one, as illustrated in Fig.~\ref{fig:overview}.

The relationship between equivalent parameters is commonly governed by an underlying symmetry, which arises naturally from the design of the synthesizer.
For example, in an additive synthesizer consisting of $k$ independent oscillators, simple permutations yield $k!$ distinct yet equivalent parameter configurations.
In this work, we focus on the effects of such permutation symmetries, which frequently occurs in synthesizers due to the use of repeated functional units --- filters, oscillators, modulation sources, etc. --- in their design.
In such cases, we show that by constructing a permutation invariant generative model from equivariant continuous normalizing flows~\cite{kohler_equivariant_2020}, we can improve over the performance of symmetry-na{\"i}ve generative models.
Further, using a toy task in which we can selectively break the invariance of the synthesizer, we show that permutation symmetry degrades the performance of regression-based models.



In real synthesizers, multiple symmetries may act concurrently on different parameters, while some parameters remain unaffected.
Hand-designing a model to achieve the appropriate invariance thus scales poorly with synthesizer complexity and requires \textit{a priori} knowledge of the underlying symmetries.
Further, some synthesizers exhibit \textit{conditional} and \textit{approximate} symmetries, for which full invariance would be overly restrictive.
To address this, we introduce a learnable mapping from synthesizer parameters to model tokens, which is capable of discovering symmetries present in the data, but can break the symmetry where necessary.
Applying this technique to a dataset sampled from the Surge XT synthesizer with multiple symmetries, continuous and discrete parameters, and audio effects, we find consistently improved audio reconstruction performance.
%
We provide full source code and audio examples at the following URL: \url{https://benhayes.net/synth-perm/}

\section{Background}\label{sec:background}

\subsection{Synthesizer inversion \& sound matching}

Given an audio signal, the sound matching task aims to find a synthesizer parameter configuration that best approximates it~\cite{roth_comparison_2011,shier_synthesizer_2021}.
%
We focus in this paper on \textit{synthesizer inversion}, a sub-task of sound matching in which the audio signal we seek to approximate is known \textit{a priori} to have come from the synthesizer.
We do so to eliminate confounding factors due to the non-trivial implicit projection from general audio signals to the set of signals producible by the synthesizer.

For an overview of historical sound matching approaches, we refer the reader to Shier's~\cite{shier_synthesizer_2021} comprehensive review.
The state-of-the-art in regression-based approaches was recently proposed by Bruford et al~\cite{bruford_synthesizer_2024}, who proposed to adopt the audio spectrogram transformer~\cite{gong_ast_2021} architecture.
Given its superior performance over MLP and CNN models, we adopt this model as our regression baseline.

Esling et al~\cite{esling_flow_2020} presented the first generative approach, which was subsequently extended by Le Vaillant et al~\cite{vaillant_improving_2021}.
These methods train variational autoencoders on audio spectrograms, enriching the posterior distribution with normalizing flows.
A second flow is jointly trained with a regression loss to predict synthesizer parameters from this learned audio representation.



Differentiable digital signal processing (DDSP)~\cite{engel_ddsp_2020,hayes_review_2023} has also been applied to sound matching~\cite{masuda_synthesizer_2021,masuda_improving_2023,uzrad_diffmoog_2024,yang_white_2023,han_learning_2024}.
Such approaches are effectively regression-based, as the composition of a differentiable synthesizer and an audio-domain loss function is a parameter-domain loss function.
If the synthesizer exhibits an invariance, so will this composed loss function, meaning DDSP-based methods are also subject to the responsibility problem.
Thus, while we do not conduct specific DDSP experimentation, we expect our findings for permutation invariant loss functions to be applicable also to DDSP with permutation invariant synthesizers.

\subsection{Permutation symmetry \& set generation}

Predicting set-structured data (such as the parameters of a permutation invariant synthesizer) with vector-valued neural networks leads to a pathology known as the \textit{responsibility problem}~\cite{hayes_responsibility_2023,zhang_fspool_2019}, in which the continuous model must learn a highly discontinuous function.
Intuitively, it is always possible to find two similar inputs that induce a change in ``responsibility'', and hence an approximation to a discontinuity in the network's outputs.
Despite these issues, in synthesizer and audio processor inference tasks it is common to ignore the symmetry at the architectural level and simply adopt permutation invariant loss functions~\cite{engel_self-supervised_2020} or symmetry breaking heuristics~\cite{nercessian_neural_2020,masuda_improving_2023}.
However, such approaches are still subject to the responsibility problem, and thus do not solve the underlying issue.

A variety of methods have been proposed for set prediction, of which the most successful view the task generatively~\cite{zhang_set_2021,zhang_deep_2019,kosiorek_conditional_2020,kim_setvae_2021} by transforming an exchangeable sample to the target set.
Effectively, the task is framed as conditional density estimation over the space of sets~\cite{zhang_set_2021}.
Based on this insight, more general generative models such as continuous normalizing flows~\cite{bilos_scalable_2021} and diffusion models~\cite{wu_fast_2023} have been adapted to permutation invariant densities.
%
%

\subsection{Continuous normalizing flows}

Continuous normalizing flows (CNFs)~\cite{chen_neural_2018,grathwohl_ffjord_2018} are a family of powerful generative models which define invertible, continuous transformations between probability distributions.
The conditional flow matching~\cite{tong_improving_2024} framework allows us to train CNFs without expensive numerical integration by sampling a \textit{conditional} probability path and regressing a closed form vector field which, in expectation, recovers the exact same gradients as regressing over the marginal field~\cite{lipman_flow_2023,tong_improving_2024}.
In this work, we adopt the rectified flow~\cite{liu_flow_2023} probability path which we pair with a minibatch approximation to the optimal transport coupling~\cite{pooladian_multisample_2023}.
We build on prior work on equivariant flows~\cite{song_equivariant_2023,hassan_equivariant_2024,klein_equivariant_2023} which are known to produce samples from invariant distributions~\cite{kohler_equivariant_2020}.


\section{Method}\label{sec:method}


Let $\calP\subset\bbR^k$ be the space of synthesizer parameters\footnote{We include MIDI pitch and note on/off times in our definition of synthesizer parameters. Effectively, we are dealing in single notes.} and $\calS\subset\bbR^n$ be the space of audio signals. A synthesizer is a map between these spaces, $f: \calP \to \calS$.
It is common that $f$ is not injective.
That is, there exist multiple sets of parameters, e.g. $\mathbf{x}^{(1)}, \mathbf{x}^{(2)} \in \calP$, that produce the same signal, i.e. $f(\mathbf{x}^{(1)})=f(\mathbf{x}^{(2)})$.
A trivial example is given when the synthesizer has a global gain control ---
all sets of parameters with zero global gain will produce an equivalent, silent signal.
Clearly, in such cases, $f$ lacks a well-defined inverse.
Therefore, we model our uncertainty over $\mathbf{x}$ as $p(\mathbf{x} \mid \mathbf{y})$, the distribution of parameters $\mathbf{x}\in\calP$ given a signal $\mathbf{y}\in\calS$.

%
%

When there is some transformation --- let us denote this $g$ --- that acts on parameters such that $f(g\cdot \x) = f(\x)$ for all $\x \in \calP$, we say that $g$ is a symmetry of $f$.
For example, $g$ might permute the oscillators.
If a set $\calG$ of such transformations contains an identity transformation and an inverse for each $g\in\calG$, then under composition it is a group acting on $\calP$.
For any $\x\in\calP$, its $\calG$-orbit --- the set of points reachable via actions $g \in\calG$ --- is defined as $O_\x = \{g\cdot \x : g \in \calG\}$.
The set of all $\calG$-orbits is a disjoint partition of $\calP$:

\begin{equation}
  \calP = \bigsqcup_{O \in \calP / \calG} O,
  \label{eq:orbit-partition}
\end{equation}

\noindent 
This allows any parameter in $\calP$ to be expressed as $\x = g \cdot r_O$ for some orbit representative $r_O \in O$ and group element $g \in \calG$.\footnote{%
  This factorization is unique if and only if the stabilizer of $\x$ (i.e. the subgroup of $\calG$ that leaves $\x$ unchanged) is trivial.
}
Hence, we can factor our conditional parameter density as:\footnote{A full derivation is given in Appendix~\ref{apdx:deriv}.}


\begin{equation}
  \begin{aligned}
    p(\x \mid \y) 
    &= 
    \underbrace{p(O \mid \y)}_\text{Orbit}\cdot
    \underbrace{p(g \mid O, \y)}_\text{Symmetry} \cdot
    \underbrace{\eta(O)}_\text{Stabilizer},
  \end{aligned}
  \label{eq:decomp}
\end{equation}

\noindent where $p(O \mid \y)$ is invariant under transformations in $\calG$, $p(g \mid O, \y)$ describes the remaining uncertainty due to symmetry, and $\eta(O)$ is a scaling factor determined by the stabilizer
of $r_O$ in $\calG$.
Depending on the nature of $\calG$, the posterior over orbits $p(O \mid \y)$ may be considerably simpler than that over parameters $p(\x \mid \y)$.
For an additive synthesizer of 16 permutable oscillators, $|\calG| = 16! \approx 20.9\times 10^{12}$, meaning that any mode of $p(\x \mid \y)$ is accompanied by over 20 trillion symmetries, while it is represented only once in $p(O \mid \y)$.

We therefore want to factor out the effect of symmetry.
Under two reasonable assumptions, it can be shown that $p(g \mid O, \y)$ is necessarily uniform over $\calG$.
First, we assume $\calG$-invariance of the likelihood $p(\y \mid \x)$ which arises naturally from the symmetry of our synthesizer.
Secondly, we assume $\calG$-invariance of the prior $p(\x)$.
This is a stronger assumption, which we satisfy by randomly sampling our training data from $\calG$-invariant parameter distributions, in contrast to some previous work which produces training data from handmade synthesizer presets.

Under a uniform symmetry distribution, we can say that $p(\x \mid \y) \propto p(O \mid \y)\eta(O)$ and reduce our task to density estimation over orbits.
Of course, the orbit of a point in $\calP$ is an abstract equivalence class and can not directly be represented.
However, by enforcing $\calG$-invariance we ensure that our model is unable to discriminate between points on the same orbit, and thus implicitly learn the orbital posterior.

\begin{figure}
\centerline{\includegraphics[width=\columnwidth]{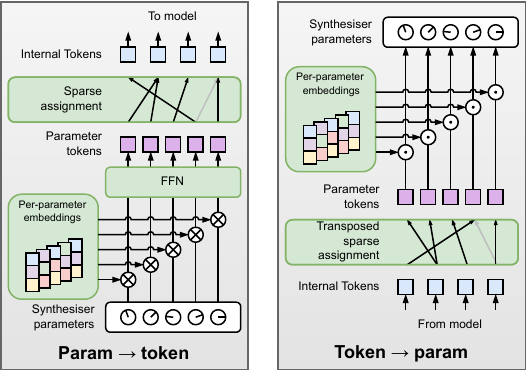}}
  \caption{The \ptot{} projection for learning to assign parameters to tokens with relaxed equivariance.}\label{fig:ptot}
\vspace{0pt plus 0pt minus 8pt}  
\end{figure}

\subsection{Permutation equivariant continuous normalizing flows}\label{sec:permflow}

Our task, then, is to learn a $\calG$-invariant distribution, focusing on the case where $\calG$ is the product of permutation (symmetric) subgroups, i.e. $\calG = \mathop{\bigtimes}_i S_{k_i}$ for orders $k_i$.
K{\"o}hler et al~\cite{kohler_equivariant_2020} showed that the pushforward of an isotropic Gaussian under an equivariant continuous normalizing flow is a density with the corresponding invariance.
We thus seek a permutation \textit{equivariant} architecture, making a Transformer~\cite{vaswani_attention_2017} encoder without positional encoding~\cite{kosiorek_conditional_2020} a natural choice.
We adopt the Diffusion Transformer (DiT)~\cite{peebles_scalable_2023} architecture with Adaptive Layer Norm (Ada-LN) conditioning.

Next, we must select an appropriate map from vectors in $\calP$ to permutable Transformer tokens.
For a simple synthesizer consisting of $k$ permutable oscillators we could simply define $k$ tokens, assigning the parameters of each oscillator to a distinct token.
However, for a synthesizer with multiple permutation symmetries, each acting on a distinct subset of parameters, and some further non-permutable parameters, this is more challenging.
As well as assigning parameters to tokens, we must indicate which tokens may be permuted with which other tokens and which may not be permuted at all.

Further, we may encounter quasi-symmetric structure in real synthesizers.\footnote{%
While formal treatment is beyond the scope of this paper, we briefly illustrate them here as they nonetheless represent a source of uncertainty in inverting real world synthesizers.
}
Specifically, we define \textit{conditional symmetry} to mean a symmetry that acts only on some subset $\calP^\prime \subset \calP$.
For example,  in Surge XT (see section~\ref{sec:surge}), only certain values of the ``routing'' parameter allow permutation symmetry between filters.
Further, if the actions of a group leave the signal \textit{almost} unchanged, up to some error bound, we call this an \textit{approximate symmetry}.
For example, swapping two filters placed before and after a soft waveshaper may yield perceptually similar, but not identical signals.
Whilst it may be possible to hand design tokenization strategies to accommodate these behaviours on a case-by-case basis, ideally we would learn such a mapping directly from data.

\subsection{Equivariance discovery with \ptot{}}\label{sec:param2tok}

To address this, we propose \ptot{}, a method to map from parameters to tokens such that the Transformer's equivariance can be used if a permutation symmetry is present in the data, but will not be enforced if not.
This is illustrated in Fig.~\ref{fig:ptot}.

Given a synthesizer with $k$ parameters, we learn a matrix $\Z\in\bbR^{k\times d}$, and represent each parameter by scalar multiplication with the corresponding row of $\Z$, followed by a feed-forward neural network $h_\theta$ applied row-wise.
To map parameter tokens to Transformer tokens, we learn an assignment matrix $\A\in\mathbb{R}^{n\times k}$, from $k$ parameters to $n$ tokens, giving the input to the first layer of our Transformer:

\begin{equation}
  \mathbf{X}_0 = \A h_\theta\left(\text{diag}(\x) \Z\right).
\end{equation}

\noindent
To return from tokens to parameters, we use weight-tying \cite{press_using_2017} of the assignment matrix and another set of learned vectors $\Z^\prime \in \mathbb{R}^{k \times d}$ as follows:

\begin{equation}
  \tilde{\x} = \left(\Z^\prime \odot \left(\A^T \X_l\right)\right)\mathbf{1}_d,
\end{equation}

\noindent where $\mathbf{1}_d$ is simply a $d$-dimensional vector of ones and $\X_l$ is the output of the $l^\text{th}$ Transformer layer.

In practice, we find it necessary to initialize $\Z, \Z^\prime,$ and $\A$ such that \ptot{} is approximately invariant to any permutation of the parameter vector.
Intuitively, the model starts with the maximum possible symmetry which is then gradually broken through training.
Finally, we encourage a sparse assignment with an $\mathcal{L}_1$ penalty on the entries of $\mathbf{A}$.
We offer regularization and initialization details in Appendix~\ref{apdx:ptot}, as well as proofs that \ptot{} can represent quasi-symmetries.

\section{Experiments}\label{sec:experiments}


We train CNF models with the rectified flow probability path~\cite{liu_flow_2023}. A minibatch approximation to the optimal transport coupling~\cite{pooladian_multisample_2023} is implemented using the Hungarian algorithm.
Conditioning dropout is applied with a probability of 10\% and inference is performed with classifier-free guidance~\cite{ho_classifier-free_2021,zheng_guided_2023} with a scale of 2.0.
Sampling is performed for 100 steps with the 4$^\text{th}$ order Runge-Kutta method.
For brevity, we omit some details about our models, datasets, and training configurations, but these can be found in Appendix~\ref{apdx:train} as well as plots of learned {\sc Param2Tok} parameters.

\subsection{Isolating permutation symmetry with $k$-osc}\label{sec:kosc}
\begin{figure*}
  \begin{center}
    \includegraphics[width=\textwidth]{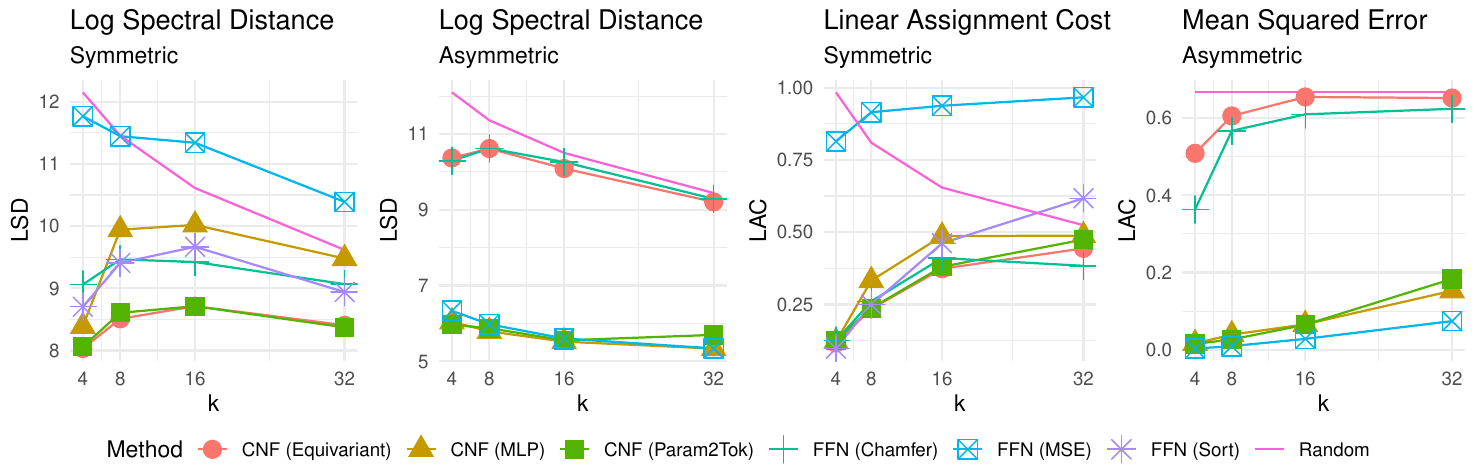}
  \caption{
    Evaluation results for the symmetric and asymmetric variants of the $k$-osc task.
  }\label{fig:kosc}
  \end{center}
\end{figure*}

To isolate the effect of permutation symmetry, we propose a simple synthetic task called $k$-osc, in which data is generated with a simple synthesizer which sums the output of $k$ identical oscillators.
Each oscillator is parameterized by: angular frequencies $\bm{\omega}$, amplitudes $\bm{\alpha}$, and waveform shapes $\bm{\gamma}$.
The waveform shape parameter linearly interpolates between sine, square, and sawtooth waves, of which the latter two are antialiased with PolyBLEP~\cite{valimaki_perceptually_2012}.

This synthesizer is permutation invariant, but we can break the symmetry by constraining each oscillator to a disjoint frequency range.
This allows us to compare the performance of models under both symmetric and asymmetric conditions.
Similarly, we can control the degree of symmetry by varying $k$, which we test at $k\in\{4,8,16,32\}$.

Data points are created by sampling a parameter vector $x = \begin{bmatrix} \bm{\omega} &  \bm{\alpha} & \bm{\gamma} \end{bmatrix}\in[-1, 1]^{3k}$.
Frequency parameters are rescaled internally to $[0, \pi]$ for the symmetric variant or $[(i - 1)\frac{\pi}{k}, i\frac{\pi}{k}]$ for $i = 1, \dots, k$ for the asymmetric variant.

\subsubsection{Models}

We compare several variants of both generative and regression-based models.
To help better interpret results, we also compute metrics over randomly sampled parameter vectors.
All models extract a representation of the signal using a 1D frequency domain CNN.

\noindent\textbf{CNF models: } We train three CNF models. The first, \cnfe{}, parameterizes the flow's vector field with a DiT and Ada-LN conditioning, with the parameter-to-token mapping described in Section~\ref{sec:permflow}. Specifically, the $i^\text{th}$ input token is $\begin{bmatrix}\omega_i & \alpha_i & \gamma_i\end{bmatrix}$, and $n=k$.
Conversely, \cnfp{} learns the mapping via the relaxed equivariance of \ptot{}, again with $k$ model tokens.
Finally, to isolate the effect of model equivariance, we train \cnfm{} using a residual MLP vector field with no intrinsic equivariance.

\noindent\textbf{Regression models: } All regression models use a residual MLP prediction head acting on the CNN representation. The first, \ffnm{} simply regresses parameters with MSE loss. The remaining models mirror approaches to permutation symmetry found in the literature~\cite{masuda_improving_2023,nercessian_neural_2020,engel_self-supervised_2020}. \ffns{} sorts oscillators by their frequency before computing the MSE loss, while \ffnc{} computes a permutation invariant loss using a differentiable Chamfer distance~\cite{barrow_parametric_1977}.
While both approaches have improved performance in prior work, neither is able to resolve the responsibility problem~\cite{hayes_responsibility_2023}.

\subsubsection{Metrics}\label{sec:kosc_metrics}

For every inferred parameter vector, we reconstruct the signal by synthesizing with these parameters.
We measure dissimilarity from the ground truth signal using the $\mathcal{L}_2$ \textit{log spectral distance} (LSD).
To compare inferred parameters with the ground truth, we simply use the mean squared error (MSE) for the asymmetric condition.
For the symmetric condition, we compute the optimal \textit{linear assignment cost} (LAC) --- that is, the minimum MSE across permutations. 
%
%
%
%

\subsubsection{Results}

Results are shown in Fig.~\ref{fig:kosc}.
Comparing \cnfe{} and \ffnm{} clearly illustrates the effect of permutation symmetry.
Under symmetric conditions, \cnfe{} excels, while \ffnm{} performs poorly across metrics.
Without symmetry, the roles are reversed --- \ffnm{} achieves excellent results, while \cnfe{} imposes an overly restrictive equivariance and fails to discriminate between oscillators. 

\ffnc{}, \ffns{}, and \cnfm{} offer ``partial'' solutions, being less affected when symmetry is present, but still underperforming \cnfe{}.
Under asymmetric conditions, \ffnc{} also imposes too stringent a restriction and performs similarly to \cnfe{}.
\cnfm{}, however, lacks explicit equivariance and performs on par with \ffnm{}.
\cnfp{} performs on par with the best models across conditions, demonstrating that \ptot{} can take advantage of model equivariance when helpful, but still learn distinct oscillator representations where necessary.
%
Broadly, both parameter- and audio-domain metrics suggest that failing to account for permutation symmetry will harm performance and that CNFs with relaxed equivariance offer good performance across both symmetric and asymmetric scenarios.

\subsection{Real-world synthesizer inversion}\label{sec:surge}
\begin{table*}[t!]
  \centering
  \renewcommand{\cellalign}{bl}
  \renewcommand{\theadalign}{bl}
  \resizebox{\textwidth}{!}{%
  \begin{tabular}{l l rrrr rrrr}
    \toprule
    & & \multicolumn{4}{c}{{\sc Surge XT} (Simple)} & \multicolumn{4}{c}{{\sc Surge XT} (Full)} \\
    Method  &  Params  &  
    \multicolumn{1}{c}{MSS $\downarrow$} & \multicolumn{1}{c}{wMFCC $\downarrow$} & \multicolumn{1}{c}{SOT $\downarrow$} & \multicolumn{1}{c}{RMS $\uparrow$} &
    \multicolumn{1}{c}{MSS $\downarrow$} & \multicolumn{1}{c}{wMFCC $\downarrow$} & \multicolumn{1}{c}{SOT $\downarrow$} & \multicolumn{1}{c}{RMS $\uparrow$}
    \\
    \cmidrule(lr){1-2}
    \cmidrule(lr){3-6}
    \cmidrule(lr){7-10}
    \cnfp{} & 40.2M &
    \textbf{3.18}	\ci{0.06}	&	\textbf{5.57}	\ci{0.06}	&	\textbf{0.042}	\ci{0.001}	&	\textbf{0.948}	\ci{0.002}	&
    \textbf{6.13}	\ci{0.09}	&	\textbf{9.63}	\ci{0.08}	&	\textbf{0.086}	\ci{0.001}	&	\textbf{0.939}	\ci{0.002}	\\
    \cnfm{}  & 41.6M & 
    3.53	\ci{0.06}	&	6.02	\ci{0.07}	&	0.045	\ci{0.001}	&	0.932	\ci{0.002}	&
    7.35	\ci{0.11}	&	11.01	\ci{0.09}	&	0.095	\ci{0.001}	&	0.924	\ci{0.002}	\\		
    \st{}  & 45.3M &
    6.51	\ci{0.10}	&	9.11	\ci{0.09}	&	0.076	\ci{0.001}	&	0.738	\ci{0.005}	&
    14.73	\ci{0.14}	&	14.42	\ci{0.11}	&	0.136	\ci{0.002}	&	0.804	\ci{0.004}	\\		
    \vae{} & 39.5M / 44.3M$^\ast$ &
    26.23	\ci{0.14}	&	15.80	\ci{0.10}	&	0.184	\ci{0.002}	&	0.595	\ci{0.006}	&
    28.86	\ci{0.19}	&	18.61	\ci{0.12}	&	0.178	\ci{0.002}	&	0.703	\ci{0.004}	\\		
    \midrule
    & &
    \multicolumn{4}{c}{Surge Full $\rightarrow$ NSynth} &
    \multicolumn{4}{c}{Surge Full $\rightarrow$ FSD50K} 
    \\
      &   &
    {MSS $\downarrow$} & {wMFCC $\downarrow$} & {SOT $\downarrow$} & {RMS $\uparrow$} &
    {MSS $\downarrow$} & {wMFCC $\downarrow$} & {SOT $\downarrow$} & {RMS $\uparrow$}
    \\
    \cmidrule(lr){1-2}
    \cmidrule(lr){3-6}
    \cmidrule(lr){7-10}
    \cnfp{} & 40.2M &
    \textbf{11.04}	\ci{0.28}	&	\textbf{19.71}	\ci{0.24}	&	\textbf{0.158}	\ci{0.004}	&	0.834	\ci{0.006}	&
    \textbf{15.40}	\ci{0.17}	&	\textbf{17.25}	\ci{0.12}	&	\textbf{0.168}	\ci{0.003}	&	0.680	\ci{0.005}	\\		
    \cnfm{}  & 41.6M &
    13.51	\ci{0.32}	&	21.09	\ci{0.23}	&	0.175	\ci{0.004}	&	\textbf{0.840}	\ci{0.006}	&
    16.96	\ci{0.18}	&	17.83	\ci{0.12}	&	0.183	\ci{0.003}	&	\textbf{0.710}	\ci{0.004}	\\		
    \st{}  & 45.3M &
    24.04	\ci{0.30}	&	28.62	\ci{0.17}	&	0.224	\ci{0.003}	&	0.755	\ci{0.007}	&
    19.09	\ci{0.16}	&	21.35	\ci{0.14}	&	0.187	\ci{0.003}	&	0.682	\ci{0.004}	\\		
    \vae{} & 44.3M &
    35.29	\ci{0.26}	&	23.50	\ci{0.14}	&	0.266	\ci{0.004}	&	0.689	\ci{0.007} &
    25.06	\ci{0.20}	&	21.24	\ci{0.12}	&	0.247	\ci{0.003}	&	0.681	\ci{0.004}	\\
    \bottomrule\rule{0pt}{0ex}\\
  \end{tabular}
  }\\\vspace{-6pt}
  \raggedright
  \footnotesize{\tiny $^\ast$ \vae{} parameter counts for \textit{Simple} and \textit{Full} datasets, respectively. The difference arises because the latent space dimension is set to the length of the parameter vector.}\\
  \centering
  \caption{\textbf{Top: } Audio reconstruction results on \textit{Simple} and \textit{Full} variants of the Surge XT synthesizer inversion task. \textbf{Bottom: } Out-of-domain audio reconstruction results for models trained on the \textit{Surge Full} dataset. \textbf{All: } Results reported with 95\% confidence interval computed across test dataset.}
  \label{tab:surge}
  \vspace{-0.5\baselineskip}
  \end{table*}

Our goal, of course, is to invert real-world software synthesizers.
To this end, we test our method on Surge XT, an open source, feature rich ``hybrid'' synthesizer which incorporates multiple methods of producing and shaping sound.
It exhibits multiple permutation symmetries, e.g. between oscillators, LFOs, and filters, as well as further conditional and approximate symmetries.

We construct two datasets of 2 million samples each by randomly sampling from the synthesizer's parameter space and rendering the corresponding audio using the {\tt pedalboard} library~\cite{sobot_pedalboard_2023}.\footnote{%
  Surge XT actually provides comprehensive Python bindings, allowing for direct programmatic interaction with the synthesizer.
  We opt to use the {\tt pedalboard} in order to ensure our system is applicable to any software synthesizer, which limits us only to parameters exposed to the plugin host.
}

The first dataset, referred to as \textit{Surge Simple}, has 92 parameters including controls for three oscillators, two filters, three envelope generators, 5 low frequency oscillators (LFOs), and some global parameters.
Omissions include discrete parameters, and those which affect global signal routing.
The second, \textit{Surge Full}, has 165 parameters, of which many are discrete, some alter the internal routing, and some control audio effects.
\textit{Surge Full} thus covers a broader sonic range, and introduces uncertainty beyond permutation symmetry.

In both cases, audio was rendered in stereo at 44.1kHz for a duration of 4 seconds and converted to a Mel-spectrogram with 128 Mel bands, using an analysis window of 25ms and a hop length of 10ms.
Spectrograms were standardized using train dataset statistics.
Like Le Vaillant et al~\cite{vaillant_improving_2021}, we adopt a one-hot representation of discrete and categorical parameters, concatenating all scalar and one-hot parameter representations to a single vector.
All synthesis parameters are scaled to the interval $[-1, 1]$.

\subsubsection{Models}

\textbf{CNF models: } All CNF models receive audio conditioning from an Audio Spectrogram Transformer (AST)~\cite{gong_ast_2021} with pre-normalization~\cite{xiong_layer_2020}.
Instead of a single {\tt [CLS]} token, we increase conditioning expressivity by learning an individual query token for each layer of the CNF's vector field.
All models are trained end-to-end.
Following prior findings on mixed continuous-discrete flow-based models~\cite{dunn_mixed_2024}, we do not constrain flows over discrete parameters to the simplex and simply adopt the same Gaussian source distribution for all dimensions.
Our \cnfp{} model again uses the \ptot{} module with a DiT vector field. Again, we train a \cnfm{} model with a residual MLP vector field to help isolate the effect of model equivariance from that of the generative approach.

\noindent\textbf{Baselines: }
We adopt the \st{}~\cite{gong_ast_2021} approach proposed by Bruford et al~\cite{bruford_synthesizer_2024} as our regression baseline, and Le Vaillant et al's~\cite{vaillant_improving_2021} \vae{} method as a further generative baseline.

\subsubsection{Metrics}\label{sec:surge_metrics}

While previous sound matching work has typically included parameter domain distances as an evaluation metric, here we argue that if the synthesizer exhibits symmetry, such metrics are uninformative and possibly misleading.\footnote{%
For more detail on this point, see Appendix~\ref{apdx:metrics}.
}
While under simple symmetries we may select an invariant metric, as in Section~\ref{sec:kosc_metrics}, this approach does not scale to more complex synthesizers.
We thus rely on audio reconstruction metrics which both share the exact invariances of the synthesizer and quantify performance on the task we are actually concerned with, i.e. best reconstructing the signal.

We measure spectrotemporal dissimilarity with a multi-scale spectral (MSS) distance computed over log-scaled Mel spectrograms at multiple resolutions.
However, a slight error in one parameter may lead an otherwise good match to be unfairly penalized due to shifts in time or frequency.
We thus include a \textit{warped Mel-frequency cepstral coefficient} (wMFCC) metric as a more ``malleable'' alternative, given by the optimal $\mathcal{L}_1$ dynamic time warping (DTW) cost between two MFCC series.
This allows robustness to timing and pitch deviations.

Pointwise spectral distances fail to capture distances ``along'' the frequency axis, such as pitch dissimilarity~\cite{turian_im_2020,hayes_responsibility_2023}.
We thus adopt the \textit{spectral optimal transport} (SOT) distance~\cite{torres_unsupervised_2024} as a pitch-sensitive measure of spectral similarity.
Finally, we include a simple cosine similarity between framewise RMS energy to capture amplitude envelope similarity.

\subsubsection{Results}

Metrics are computed on a held-out test dataset of $10,000$ sounds synthesized in the same manner as the training dataset.
Results are presented in Table~\ref{tab:surge}, and audio examples can be heard on the companion website.\footnote{Link to website:~\url{https://benhayes.net/synth-perm/}}

The CNF models perform consistently with our theoretical expectations and the results of our toy task.
In particular, \cnfp{} outperforms all models across all metrics on both datasets.
The learned \ptot{} parameters suggest that the model has discovered a token mapping that corresponds to the intrinsic symmetries of the synthesizer, and thus benefits from the Transformer's permutation equivariance.

The non-equivariant \cnfm{} also achieves reasonable performance, suggesting that simply adopting a probabilistic framework is already very helpful in dealing with the sources of ill-posedness in real synthesizers.
The difference between the MLP and \ptot{} variants is more pronounced for \textit{Surge Full}, suggesting that the effect of introducing greater complexity is magnified in the presence of unresolved symmetry.
This aligns with our decomposition of the conditional parameter density in Section~\ref{sec:method} --- any change in $p(O \mid \y)$ is ``repeated'' in $p(\x \mid \y)$ for each element of $\calG$.

Despite extensive tuning, \vae{} consistently collapsed to predicting average values for many parameters.
Of course, this conflicts with the results of the original papers~\cite{vaillant_improving_2021,esling_flow_2020}.
We suggest this may be due to the use of smaller datasets of hand-designed synthesizer presets, which could be sufficiently biased to break the invariance of $p(\x)$.

%
%
\subsubsection{Out-of-distribution results}

While this is not our focus, we report out-of-distribution results on audio from the NSynth~\cite{engel_neural_2017} and FSD50K~\cite{fonseca_fsd50k_2022} test datasets in Table~\ref{tab:surge}.
Across all models, performance suffers compared to in-domain data, but the relationships between models are preserved.
Thus, even under such challenging conditions, unhandled symmetry likely detrimentally influences performance.
To adapt our method to general sound matching, we intend to explore shared representations of synthesized and non-synthesized audio, domain adaptation techniques, and information bottlenecks, such as quantization, for conditioning.


\section{Conclusion}\label{sec:conclusion}



%

The implications of our findings are clear: if the synthesizer has a symmetry, it is better to (i) approach the problem generatively and (ii) learn the corresponding invariant density.
This extends beyond synthesizers, as audio effects also commonly exhibit permutation symmetries, as noted by Nercessian~\cite{nercessian_neural_2020}.
Beyond audio, these results are of relevance anywhere neural networks parameterize an external system with structural symmetries.

%

A key limitation of our work is the lack of specific experimentation on the role of quasi-symmetries.
A $k$-osc style task which isolates their effect would be very illuminating about the extent to which \ptot{} improves results in their presence, and we intend to include this in a subsequent publication.
We also note that we currently lack theoretical guarantees that \ptot{} will discover symmetries.
In future work, we thus intend to both study this module theoretically and gather further empirical data on the role of initialization and regularization in its behaviour.
We also plan more comprehensive evaluation of the extent to which \ptot{}-based models learn the appropriate invariance.



Our proposed method should generalize effectively to arbitrary software synthesizers.
In future work, we will therefore explore multi-task training by simultaneously modelling over many synthesizers. 
We will also seek to extend our work to the more general sound matching task by exploring the previously discussed strategies for robustifying our system to out-of-distribution inputs.


\section{Ethics Statement}

Like any AI model, our work inherently encodes the biases and values of the authors.
In particular, our choice of VST synthesizer reflects a bias towards western popular music.
However, the more abstract nature of our synthetic experimentation does suggest that our results may reasonably be expected to generalize to tools that better represent the needs of other musical cultures.

The source of training data is of particular importance in assessing the ethical impact of an AI model, and in this instance we have worked entirely with synthetically generated data.
Even in our experimentation on a real synthesizer, we opted to generate our dataset by randomly sampling parameters, rather than by scraping the internet for presets or exploiting factory preset libraries.
We further note that Surge XT is an open source synthesizer released under the GNU GPL.

Amid growing concern over AI tools displacing workers, particularly in the creative industries, we wish to highlight that our goal in this work is to develop technology to integrate with and enhance existing workflows.
We believe the choice of task reflects this by seeking to enhance interaction with an existing family of creative tools, as opposed to outright replacing them.

\section{Acknowledgements}

B.H. would like to thank Christopher Mitcheltree, Marco Pasini, Chin-Yun Yu, Jack Loth, and Julien Guinot for their invaluable feedback on this manuscript in varying stages of completion, and Jordie Shier for the many inspiring and illuminating conversations on this topic.

This research was supported by UK Research and Innovation [grant number EP/S022694/1]. It utilised Queen Mary’s Apocrita HPC facility, supported by QMUL Research-IT. \url{http://doi.org/10.5281/zenodo.438045}

\defbibheading{bibliography}[References]{%
  \section{#1}
}
\printbibliography

\appendix
\section{Auxiliary derivations \& proofs}\label{apdx:deriv}

\subsection{Orbital factorization of $p(\x \mid \y)$}

Recall that $\calP\subset\bbR^k$ is the space of synthesizer parameters, acted on by $\calG$.
The $\calG$-orbit of a point $\x\in\calP$ is the set of all points in $\calP$ that are reachable from $\x$ by actions $g\in\calG$. That is:

\begin{equation}
  O_\x = \left\{
    g\cdot \x  : g \in \calG
    \right\}.
\end{equation}

It is well known that the orbits of a group action on a set form a disjoint partition of the set, which we denote $\calP \mathbin{/} \calG$.
We select a representative in $\calP$ for each orbit, $O \longleftrightarrow r_O$, which by the partitioning of $\calP$ is unique.

Then for any $\x\in\calP$ there exists a $g\in\calG$ such that $\x = g \cdot r_O$ for some unique $O$.
The uniqueness of $g$ is a slightly more nuanced matter.

\subsubsection{Uniqueness of $g$}

We start with a brief definition. The stabilizer subgroup of $\calG$ with respect to $\x$ is given by:

\begin{equation}
  \text{Stab}(\x) = \left\{
    g \in \calG : g\cdot \x = \x
  \right\}.
\end{equation}

We say $\calG$ is \textit{free} if and only if $\text{Stab}(\x) = \{ e_\calG \}$ for all $\x\in\calP$ where $e_\calG$ is the identity element of $\calG$.

Formally, in our setting, we can not safely assume that $\calG$ is free as there may exist parameter configurations with non-trivial stabilizers.
For example, if two oscillators in a multi-oscillator synthesizer have matching parameters and $\calG$ permutes oscillators, then the permutation which swaps them is clearly a member of the stabilizer subgroup.
Practically, this is of little concern as the probability of sampling a parameter vector with a non-trivial stabilizer is very low.
Nonetheless, we offer the complete derivation addressing non-trivial stabilizers here.

\subsection{Factorization of $p(\x \mid \y)$}

We wish to show that $p(\x \mid \y)$ can be decomposed $p(O \mid \y)p(g \mid O, \y)\eta(O)$.
We do so assuming $\calG$ is a finite group.
First, by the law of total probability we have:

\begin{align}
  p(\x \mid \y) &= \sum_{O\in \calP / \calG} p(\x, O \mid \y) \\
  &= p(\x, O_\x \mid \y).
\end{align}

\noindent Applying the law again, and using the representation of $\x$ as $g \cdot r_O$, we get:

\begin{align}
  p(\x, O_\x \mid \y) &= \sum_{g\in\calG} p(\x, O_\x, g \mid \y) \\
  &= \sum_{g\in\calG} p(O_\x, g \mid \y) \cdot \delta_{\x, g \cdot r_{O_\x}},
\end{align}

\noindent where $\delta_{a,b} = 1$ if $a=b$ and $0$ otherwise, i.e. the Kronecker delta.
By the orbit-stabilizer theorem, there are $|\text{Stab}(r_{O_\x})|$ unique elements in $\calG$ that satisfy $\x = g\cdot r_{O_\x}$.
Thus, we can state that:

\begin{align}
  p(\x \mid \y) &= |\text{Stab}(r_{O_\x})| \cdot p(O_\x, g \mid \y) \\
  &= p(O_\x \mid \y) \cdot p(g \mid O_\x, \y) \cdot \eta(O_\x),
\end{align}

\noindent where $\eta(O_\x) = |\text{Stab}(r_{O_\x})|$.
For trivial stabilizers, $\eta(O_\x)=1$.
For convenience, in the manuscript we adopt the convention that $O=O_\x$.

\subsection{Uniformity of $p(g \mid O, \y)$}

In the manuscript, we assert that if $p(\y \mid \x)$ and $p(\x)$ are $\calG$-invariant, then $p(g \mid O, \y)$ is uniformly distributed over $\calG$.
Our assumptions are:

\begin{equation}
  \begin{array}{rll}
    p(\y \mid \x) &= p(\y\mid g \cdot\x) & \quad \forall g\in\calG \\
    p(\x) &= p(g \cdot \x) & \quad \forall g\in\calG.
  \end{array}
\end{equation}

\noindent Then, it is straightforward to show that the posterior is also $\calG$-invariant:

\begin{align}
 p(g\cdot\x\mid\y)
  &\propto
  p(\y \mid g\cdot\x)p(g\cdot\x) \\
  &=
  p(\y\mid\x)p(\x) \\
  &\propto
  p(\x\mid\y).
\end{align}

\noindent Recalling our factorization of $p(\x\mid\y)$ and using the invariance of the posterior, we observe that for all $h\in\calG$:

\begin{equation}
  p(O \mid \y)p(g\mid O, \y) =
  p(O \mid \y)p(h\cdot g\mid O, \y),
\end{equation}

\noindent noting that $\eta(O) \geq 1$.
Therefore,

\begin{equation}
  p(g\mid O, \y) =
  p(h\cdot g\mid O, \y), \quad \forall h\in\calG.
\end{equation}

\noindent Of course, this is only true if $p(g \mid O,\y)$ is uniform over $\calG$, and so we can state, as in the manuscript, that:

\begin{equation}
p(\x\mid\y)\propto p(O\mid\y)\eta(O)
\end{equation}

\subsection{Conditional symmetry in \ptot{}}

\begin{definition}[\ptot{}]
  Let $\text{P2T}: \calP \to \bbR^{n\times d}$ be the mapping from parameters to tokens:
  \[
    \text{P2T}(x) =
      \A h_\theta\left(\text{diag}(\x)\cdot \Z\right),
      \nonumber
    \]
  for $\A\in\bbR^{n \times k}$, $\Z\in\bbR^{k\times d}$, and FFN $h_\theta$ applied row-wise.
\end{definition}

We aim to show that, in combination with the first Transformer layer, the \ptot{} block is capable of representing a conditional symmetry by implementing conditional equivariance.
That is, the equivariance can be broken depending on the value of some entry in the parameter vector $\x\in\calP$.
We do so by construction, simply to illustrate that this is possible.

For simplicity, we assume that the residual multi-head attention block implements the identity transformation, which is possible if the output projection is zero.
We also ignore the effect of layer normalization, without loss of generality.
Thus, we focus on the effect of the FFN in the first transformer layer.

Let $\x = \begin{bmatrix}a & b & c\end{bmatrix}^T \in [0, 1]^3$ be our parameter vector. Let the permutation $\pi(\x) =\begin{bmatrix}b & a & c\end{bmatrix}^T$.
Let the permutation $\rho$ act on a $2 \times d$ matrix by swapping its rows, i.e. $\rho\left(\begin{bmatrix} \mathbf{a} & \mathbf{b}\end{bmatrix}^T\right) = \begin{bmatrix} \mathbf{b} & \mathbf{a}\end{bmatrix}^T$.

Let the FFN be defined $\text{FFN}(\y) = \text{ReLU}(\y\cdot\mathbf{W}_\text{in}) \cdot \mathbf{W}_\text{out}$.

\begin{theorem}
There exist $\text{P2T}$ and $\text{FFN}$ parameters such that 

$$
\begin{array}{ll}
  c = 1: & \text{FFN}\circ\text{P2T}\circ\pi(\x) = \rho\circ\text{FFN}\circ\text{P2T}(\x) \\
  c = 0: & \text{FFN}\circ\text{P2T}\circ\pi(\x) \neq \rho\circ\text{FFN}\circ\text{P2T}(\x).
\end{array}
$$
\end{theorem}

\begin{proof}
First we construct our $\text{P2T}$ parameters. We set the parameter embeddings $\Z$ as follows:

$$
\Z =
\begin{bmatrix}
  u & v_1 & 0 \\
  u & v_2 & 0 \\
  0 & 0 & 1
\end{bmatrix},
$$

\noindent for $u, v_1, v_2 \in [0, 1]$.
Note the assumption that $a, b, c, u, v_1, v_2 \in [0, 1]$ is made without loss of generality for ease of construction.
In practice, values on intervals that include negative values can be shifted to the positive real line by bias terms.

We set $h_\theta$ such that its weights are identity matrices and assume a ReLU activation. Thus, because we have enforced positivity $h_\theta(\text{diag}(\x)\cdot\Z)=\text{diag}(\x)\cdot\Z$.

We then construct the assignment matrix $\A$:

$$
\A =
\begin{bmatrix}
  1 & 0 & 1 \\
  0 & 1 & 1
\end{bmatrix}.
$$

\noindent Then, finally, we have:

$$
\X_0 = \text{P2T}(\x) = \begin{bmatrix} 
  a u & a v_1 & c \\
  b u & b v_2 & c \\
\end{bmatrix}.
$$.

Let $\mathbf{W}_\text{out}=I_3$. Let 

$$
\mathbf{W}_\text{in}=
\begin{bmatrix}
  1 & 0 & 0 \\
  0 & 1 & 0 \\
  0 & -1 & 0
\end{bmatrix}.
$$

\noindent Then
$$
\X_0 \cdot \mathbf{W}_\text{in}= \begin{bmatrix}
  a u & a v_1 - c & 0 \\
  b u & b v_2 - c & 0
\end{bmatrix}.
$$

\noindent Then for $c=1$ we have:

$$
\mathbf{Y}_1 = \text{FFN}(\X_0)\rvert_{c=1}
= 
\begin{bmatrix}
  a u & 0 & 0 \\
  b u & 0 & 0 \\
\end{bmatrix}
$$

\noindent and for $c=0$ we have:

$$
\mathbf{Y}_0 = \text{FFN}(\X_0)\rvert_{c=0} = \begin{bmatrix}
  a u & a v_1 & 0 \\
  b u & b v_2 & 0
\end{bmatrix}.
$$

We now apply the permutation operators $\pi$ and $\rho$ to test our construction.

$$
\rho(\mathbf{Y}_1) = 
\begin{bmatrix}
  b u & 0 & 0 \\
  a u & 0 & 0 \\
\end{bmatrix}
$$

\noindent while

\begin{align}
  \nonumber
  \text{FFN}\circ\text{P2T}\circ\pi(\x)\rvert_{c=1} &=
\text{FFN}\left(
\begin{bmatrix}
  b u & b v_1 & 1 \\
  a u & a v_2 & 1 
\end{bmatrix} 
  \right) \\\nonumber
  &=
\begin{bmatrix}
  b u & 0 & 0 \\
  a u & 0 & 0 \\
\end{bmatrix}.
\end{align}

\noindent Meaning $\text{FFN}\circ\text{P2T}\circ\pi(\x) = \rho\circ\text{FFN}\circ\text{P2T}(\x)$ and we have satisfied the $c=1$ case.

Meanwhile, for $c=0$:

$$
\rho(\mathbf{Y}_0) = 
\begin{bmatrix}
  b u & b v_2 & 0 \\
  a u & a v_1 & 0 \\
\end{bmatrix}
$$

\noindent
\begin{align}
  \nonumber
  \text{FFN}\circ\text{P2T}\circ\pi(\x)\rvert_{c=0} &=
\text{FFN}\left(
\begin{bmatrix}
  b u & b v_1 & 0 \\
  a u & a v_2 & 0 
\end{bmatrix} 
  \right) \\\nonumber
  &=
\begin{bmatrix}
  b u & b v_1 & 0 \\
  a u & a v_2 & 0 \\
\end{bmatrix}.
\end{align}

\noindent Clearly, $bv_2 \neq bv_1$ and $av_2 \neq av_1$, and so we have satisfied that $\text{FFN}\circ\text{P2T}\circ\pi(\x) \neq \rho\circ\text{FFN}\circ\text{P2T}(\x)$ when $c=0$.
We have therefore constructed a conditional equivariance, so we are done.
\end{proof}

\section{Models, hyperparameters, and training details}\label{apdx:train}

\subsection{$k$-osc Models}

\begin{figure*}
  \begin{center}
    \includegraphics[width=\textwidth]{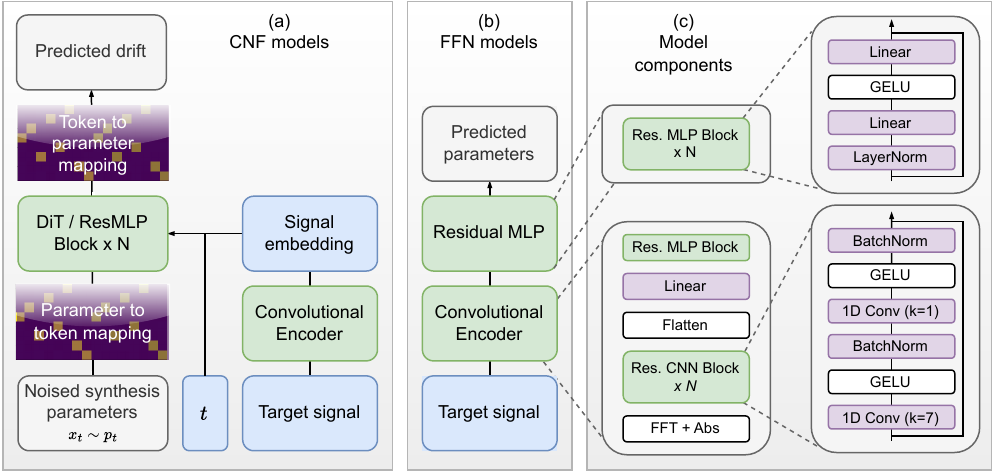}
  \end{center}
  \caption{Diagram of models used in the $k$-osc task.}\label{fig:kosc-models}
\end{figure*}

All models were trained for 400k steps using the Adam optimizer with a learning rate of $10^{-4}$. All other optimizer parameters were set to their PyTorch defaults.
Models were trained with a batch size of 1024.
All test metrics were computed on a held-out test dataset of $51,200$ samples.
Training was conducted with BFloat16 mixed precision. Gradient norms were clipped to a maximum value of 1.0.

\subsubsection{CNN Encoder}

All models used in the $k$-osc experiment use the same CNN architecture to extract a representation of the synthesizer signal.
This is illustrated in Fig.~\ref{fig:kosc-models}.

The encoder starts by computing the discrete Fourier transform of the signal and taking its complex modulus. There is then a stack of 4 residual CNN blocks, which are interspersed by a strided convolution which downsamples by a factor of 3 while doubling the channel dimension. The dimension of the first block is 24.
The final feature map is flattened and passed to a linear layer, and then a residual MLP block, producing a 128 dimensional vector which is passed to the rest of the network.


\subsubsection{CNF models}

These models are illustrated in Fig.~\ref{fig:kosc-models}(a).
\cnfp{} and \cnfe{} both use a diffusion transformer (DiT) vector field network with 5 layers.
Each layer contains two residual blocks.
The first consists of layer normalization followed by a multi-head self attention module with 4 heads and a token dimension of 128.
The second consists of a further layer normalization followed by a feed-forward network (FFN) with an input/output dimension of 128, a hidden dimension of 256, and a GELU activation.

\cnfm{} replaces the DiT with a residual multi-layer perceptron (MLP) constructed by simply omitting the attention residual blocks leaving only residual FFN layers.
Note that the MLP is applied across the full parameter vector rather than tokenwise --- that is, it does not have an intrinsic equivariance to permutation.
We use a model dimension of 192 and a total of 7 layers.

\textbf{Input projection: }
\cnfe{} uses the natural tokenization of parameters implied by the synthesizer's structure.
That is, one oscillator's parameters are assigned to each token.
These 3-dimensional vectors are then projected up to the model dimension by a linear layer.
This naturally yields $k$ internal model tokens.

\cnfp{} uses the \ptot{} projection, where the number of model tokens is a tunable hyperparameter corresponding to the ``available'' degree of symmetry.
While in future work we intend to study the effect over under- and over-specifying this parameter, here we are focused simply on the ability of the model to discover the symmetry, and so we simply set it to $k$, assuming this is known \textit{a priori}.

\cnfm{} simply projects the parameter vector up to the model dimension with a linear layer.

\textbf{Conditioning: }
There are two sources of conditioning: the CNN's output, and the flow matching time step.
The latter is represented as a simple scalar value, which is concatenated to the CNN's output, giving a 129 dimensional conditioning vector.

We use adaptive layer normalization (Ada-LN) to introduce conditioning information to the network. 
This means that after each layer normalization, a pointwise affine transformation is applied, with shift and scale parameters produced by a separate FFN per layer applied to a conditioning vector.
The output of the non-residual pathway of each residual block is also scaled by a further parameter produced by the conditioning FFN.
For \cnfp{} and \cnfe{}, this means the layer-wise conditioning FFN has a total output dimension of $6 \times 128 = 768$, which is split to produce the affine parameters.
For \cnfm{} it outputs a $3\times 128 = 384$ dimensional vector as there is only one residual block per layer.
Before being passed to each layer, the conditioning vector is also passed through a shared FFN.
Note that, unlike prior work, we find that an Ada-LN-Zero initialization (where the conditioning FFN is initialized to leave the main network's activations unchanged) actually harms performance, and so we do not use it here. 

With a probability of 10\%, we replace the CNN conditioning vector with a learned ``dropout'' vector.
Doing so allows us to use classifier-free guidance~\cite{ho_classifier-free_2021,zheng_guided_2023} at inference, which we apply with a scale of 2.0, determined through experimentation to provide reasonable performance across models.
Future work will offer a more comprehensive study of these design parameters.

\subsubsection{Regression models}

These are illustrated in Fig.~\ref{fig:kosc-models}(b).
All use the same architecture and vary only in their loss computation.

The output of the CNN is passed directly to a residual MLP.
Each residual MLP block consists, again, of a residual block containing a layer normalization, followed by a linear layer, a GELU activation, and a final linear layer before adding the residual signal.
The residual MLP head contains of 4 layers with a width of 384.
A final linear layer is used to project the MLP's output back to the dimension of the parameter vector.

\textbf{Loss functions ---}
\ffnm{} simply regresses parameters using a mean-squared error loss. This means the model does not account for symmetry at either the architectural or loss level.

\ffns{} adopts the sorting strategy for symmetry breaking adopted in prior sound matching~\cite{masuda_improving_2023} and EQ matching~\cite{nercessian_neural_2020} work.
In this approach, the target parameter vector is modified by sorting its entries according to oscillator frequency.
In this sense, the symmetry of the problem with respect to any individual parameter vector is effectively broken --- there exists a canonical, sorted form.
However, across the space of all possible parameter vectors, prior theoretical work~\cite{hayes_responsibility_2023} tells us that this does \textit{not} resolve the responsibility problem, as similar \textit{sets} of parameters can be selected that lead to big differences between parameter \textit{vectors}.
We thus expect improved performance over the \ffnm{} baseline --- there will no longer be a risk of conflicting gradients with respect to any one example --- but degraded performance compared to models which account for the symmetry at an architectural level, as the responsibility problem still exists.
We do not include the \ffns{} variant on the asymmetric version of the task, as targets are effectively already implicitly sorted.

We use \ffnc{} to study the effect of permutation invariant loss functions.
Like with \ffns{}, we expect this to improve over \ffnm{} by reducing conflicts in the learning signal, but to still suffer degraded performance due to the responsibility problem.
We use a differentiable Chamfer distance~\cite{barrow_parametric_1977} as our loss function, defined between two sets $\mathcal{X}$ and $\mathcal{Y}$ as:

\begin{equation}
  \mathcal{L}_\text{CD}(\mathcal{X}, \mathcal{Y}) =
  \sum_{\y\in \mathcal{Y}}\min_{\x\in \mathcal{X}} d(\x, \y) +
  \sum_{\x\in \mathcal{X}}\min_{\y\in \mathcal{Y}} d(\x, \y).
  \label{eq:chamfer}
\end{equation}

Here, we adopt the same mapping to set elements as \cnfe{} and define $d$ as the mean squared error.
Losses based on the Chamfer distance are commonly used as an approximation to the true linear assignment distance, which requires online solution of the corresponding combinatorial optimization problem.
A Chamfer-style loss was used by Engel, et al.~\cite{engel_self-supervised_2020} two match the frequencies of sets of sinusoids.

\subsection{Surge XT models}

We train all models for 1 million steps.
For all models except \vae{}, we use the Adam optimizer, with an initial learning rate of $10^{-4}$, which is decayed to $10^{-6}$ using the cosine learning rate schedule.
For \vae{}, we adopt an initial learning rate of $2\times 10^{-4}$. 
All models were trained with a batch size of 128 samples.
Again, we train under BFloat16 mixed precision, and clip gradient norms, though on this task we found a maximum norm of 0.5 to better ensure stability.




\subsubsection{CNF models}

\begin{figure}[t]
\centerline{\includegraphics[width=\columnwidth]{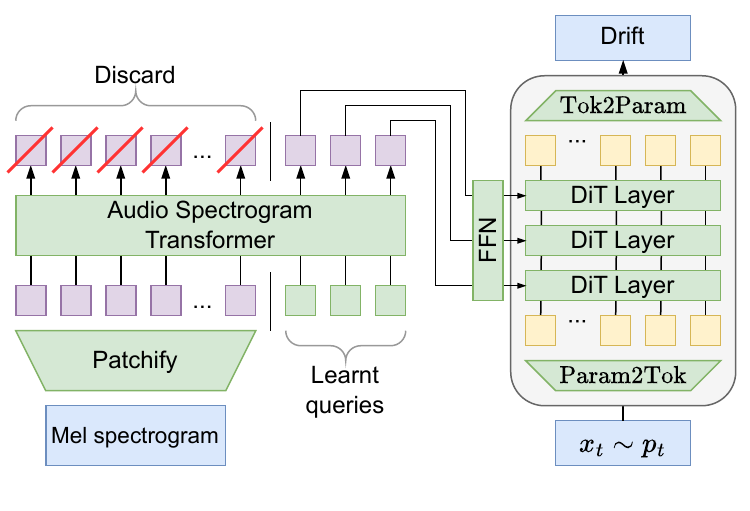}}
  \caption{We extract audio representations with an audio spectrogram Transformer~\cite{gong_ast_2021} encoder which condition an approximately equivariant continuous normalizing flow.}\label{fig:surge}
\end{figure}

The vector field architecture of the \cnfp{} and \cnfm{} models is effectively equivalent to their counterparts in the $k$-osc task, with two changes.
First, the time variable of the conditional vector field is represented via a 256 dimension sinusoidal positional encoding.
Second, model capacity is increased by increasing the model dimension and FFN hidden size to 512, using 8 attention heads and 8 total layers.
For \cnfp{}, illustrated in Fig.~\ref{fig:surge}, the number of model tokens is set to 128.

Both \cnfp{} and \cnfm{} receive conditioning from an audio spectrogram transformer~\cite{gong_ast_2021} encoder, adapted to use pre-normalization~\cite{xiong_layer_2020}.
This encoder has a model dimension of 512 and again uses 8 attention heads. There are 8 layers in the encoder.
The input is converted to patches using a CNN encoder following the original AST implementation, with a patch size of 16 and a stride of 10.

Unlike in the $k$-osc task, where each layer of the vector field received the same conditioning vector, here we produce a distinct conditioning vector for each layer by passing 8 separate learned embeddings through the AST along with its input tokens.
Outputs from these tokens are then distributed to each layer.

\subsubsection{AST model}

Our regression baseline uses an AST with an MLP prediction head, as proposed by Bruford et al~\cite{bruford_synthesizer_2024}.
It uses effectively the same architecture as the AST encoder for the CNF models, but with increased capacity.
Specifically, there are 12 layers, each with a dimension of 768 and 12 attention heads.
The prediction head adds an extra residual MLP block before a final linear projection to the output dimension.

\subsubsection{VAE + RealNVP}

We adopt the same training objectives and overall formulation as the model proposed by Le Vaillant et al~\cite{vaillant_improving_2021}.
The model consists of a variational autoencoder trained to reconstruct log-scaled Mel spectrograms.
The latent posterior distribution is enriched by a series of RealNVP~\cite{dinh_density_2017} normalizing flow layers.
As proposed by Esling et al~\cite{esling_flow_2020}, the model also regresses synthesizer parameters using a further sequence of RealNVP layers applied to the output of the latent flow layers.
However, rather than training by maximum likelihood, this ``parameter flow'' is trained by a simple regression loss --- a formulation referred to as \textit{regression flow}~\cite{esling_flow_2020}.

In order to make a fair comparison, we match the capacity of the model to our CNF and \st{} architectures.
To do so, we increase the capacity of each convolutional block by adding extra depth-wise separable convolutions between downsampling/upsampling convolution layers.
These are implemented with residual connections to help gradient flow in the deeper model.
We also replace the channel-wise convolutional network in the authors' implementation with one network that accepts a two channel (stereo) input.
Finally, we increase the width of the networks used inside the RealNVP flow layers from 300 to 512, and use 15 parameter regression flow layers instead of 6.
While these changes represent a departure from the original architecture, we find that the regression collapse we observe occurs with or without these adaptations.
In fact, they appear to slightly improve the final metric values.

For the $\beta$ parameter, which weights the contribution of the latent Kullback-Leibler divergence term in the ELBO loss, we adopt a warmup schedule similar to the one used in the original paper.
However, we proportionally extend the warmup period to correspond to the number of training steps we use.
Thus, $\beta$ increases from 0.1 to 0.2 over 60k steps.

\subsection{\ptot{} initialization}\label{apdx:ptot}

As mentioned in the manuscript, we observed empirically that \ptot{} converges much more reliably to a solution with an equivariance matching the synthesizer's symmetries when it is initialized to be \textit{almost invariant} to permutations.
Almost invariance of the assignment matrix $\A$ is preferable to full invariance, in this situation, as otherwise all parameter tokens would receive equal gradients and the model would fail to distinguish between parameters.

Our specific invariance scheme is as follows.
First, we initialize entries $a_{ij}$ of $\A$ to $a_{ij}\sim\mathcal{N}(\left(nk\right)^{-\frac{1}{2}}, \sigma^2_\A)$, where we set $\sigma^2_\A=10^{-4}$.
Then, to initialize the parameter representations $\Z$ and $\Z^\prime$, we first sample a single vector $\mathbf{m} \sim \mathcal{N}\left(0, d^{-\frac{1}{2}}I\right)$ and use this to sample $\mathbf{z}_i \sim \mathcal{N}\left(m, \sigma^2_\Z I\right)$, for $\sigma^2_\Z=10^{-4}$.
We find that first initializing $\mathbf{m}$ helps with the rate and reliability of symmetry discovery.
We hypothesize that this may be due to an interaction with layer normalization in the vector field network --- if parameter representations are centered around zero, then layer normalization will more dramatically magnify the differences between them than if they share an offset.
We set $\Z^\prime = \Z$ at the start of training, but optimize these parameters independently.

Intuitively, we wish to enforce a sparse assignment of parameter tokens to model tokens to encourage the learned equivariance to follow the symmetry structure of the synthesizer.
That is, we want to minimize unnecessary "leakage" between subsystems.
We thus adopt an $\mathcal{L}_1$ penalty on the entries in the $\A$ matrix, which we weight by a factor of 0.01.

While these initialization and regularization strategies were determined through a combination of problem insight and preliminary experimentation, we do find them necessary to achieve good performance.
We intend, in our next publication, to include detailed ablations and parameter sweeps to draw better insight into their behaviour.

\section{Metrics}\label{apdx:metrics}

Here we provide full details of metrics used in the $k$-osc and Surge XT experiments.

\subsection{$k$-osc metrics}

\textbf{LAC: }
The linear assignment cost is defined between two sets of points by finding the minimum cost attainable by assignments between them.
This is achieved by solving the linear assignment problem, typically using the Hungarian or Jonker-Volgenant algorithms.
For ground truth parameters $\x$ and predicted parameters $\hat{\x}$, and permutations of oscillators $\pi \in S_k$, the cost is defined:

\begin{equation}
  \text{LAC} = \frac{1}{3k}\min_{\pi \in S_k}  \lVert\x - \pi \hat{\x}\rVert_2^2.
  \label{eq:lac}
\end{equation}

\noindent
The resulting metric is thus a permutation invariant distance between two sets of parameters, because regardless of how parameters are permuted, the minimum cost remains unchanged.

\textbf{MSE: }
On the asymmetric variant of $k$-osc, the LAC would be an inappropriate metric as it would reward incorrect predictions.
We thus use the straightforward mean squared error between the predicted and ground truth parameter vectors in its place.

\textbf{LSD: }
The log-spectral distance measures the pointwise dissimilarity between the log-scaled magnitude spectra of two signals.
Here, given ground truth signal $\y$ and predicted signal $\hat{\y}$, we compute the discrete Fourier transforms $\Y$ and $\hat{\Y}$.
Then the LSD is given by:

\begin{equation}
  \text{LSD} = 
  \sqrt{\frac{1}{N}
  \sum_{n=1}^N
  \left(
  \log |Y_n| - \log |\hat{Y}_n|
  \right)^2
  }.
  \label{eq:lsd}
\end{equation}

\subsection{Surge XT metrics}

While previous evaluations of synthesizer inversion have included parameter-space distances, here we argue that when symmetries are present such metrics are uninformative at best and misleading at worst.
To see why, consider a ground truth parameter vector $x=\begin{bmatrix}u & v\end{bmatrix}^T$ where $u, v \in \mathbb{R}^d$ and a synthesizer that is symmetric under permutations of $u$ and $v$.
For two predictions, $x^\prime=\begin{bmatrix}v & u\end{bmatrix}$ and $x^{\prime\prime}=\frac{1}{2}\begin{bmatrix}u + v & u + v\end{bmatrix}^T$, we see that when $u\neq v$ the $\mathcal{L}_2$ distance $\lVert x - x^\prime \rVert_2 = \sqrt{2}\lVert u - v \rVert_2$ is greater than the distance $\lVert x - x^{\prime\prime} \rVert_2 = \sqrt{2}\lVert \frac{1}{2}(u - v) \rVert_2$, despite the fact that $f(x^\prime)=f(x)$ while $f(x^{\prime\prime})\neq f(x)$.
In other words, the parameter metric \textit{penalizes} a valid solution while rewarding a bad solution.

A good parameter metric would, therefore, account for such symmetries and equivalences. However, due to the complex nature of audio synthesizers and the presence of quasi-symmetric structure, it is typically unfeasible to enumerate all possible symmetries and design such a metric.
We therefore, instead, rely on the fact that the synthesizer itself already encodes all the relevant symmetries and equivalences, and simply compute metrics directly between audio signals.

While this method itself is imperfect --- parameters differ in the magnitude of their influence on audio metrics --- we argue that it is still preferable because (i) it is easier to design metrics to compensate for these differences than it is to design parameter metrics which encode the necessary invariances, and (ii) we are ultimately concerned with achieving good audio similarity, and are okay with whichever parameter configuration actually does so. We therefore adopt the following suite of metrics.

\textbf{MSS: }
Multi-scale time-frequency representations have found common use as both training loss functions and metrics for audio similarity across a variety of audio processing tasks.
They typically involve multiple evaluations of a time-frequency transform, such as the STFT, with differing filter and stride lengths, selected to achieve sensitivity to spectrotemporal structures at differing scales.
Effectively, they allow heuristic compensation for the intrinsic trade-off between time and frequency resolution in designing a transform.

Here, we use log-scaled Mel spectrograms with window lengths of 10, 25, and 100 milliseconds, and corresponding hop lengths of 5, 10, and 50ms.
We use filterbanks with 32, 64, and 128 mel bands, respectively.
We compute the mean absolute error between corresponding spectrograms, and aggregate these by averaging across scales.

\textbf{wMFCC: }
Due to the intricate interactions between synthesizer parameters, it is possible for a slight error in one parameter to shift the synthesizer's output in time or frequency, thus leading to simple spectrotemporal metrics unfairly penalizing an otherwise good prediction.
Even a few milliseconds error on an amplitude envelope, or a few cents on an oscillator frequency, could cause a signal which is almost indistinguishable perceptually from the ground truth to suffer a poor MSS value.

Thus, we include a \textit{warped Mel-frequency cepstral coefficient} metric as a more ``malleable'' alternative to straightforward spectral distances.
The metric is computed in two steps.
First, the first 20 MFCCs of both the ground truth and predicted signals are computed, with a hop length of 10ms.
As MFCCs are more sensitive to the spectral envelope than the absolute pitch of a signal, this allows for greater robustness to small pitch errors.
Secondly, the dynamic time warping algorithm is applied to the two MFCC series using the L1 distance, allowing for small timing deviations to be accommodated.
The cost of the optimal DTW path is then taken as the wMFCC value.

\textbf{SOT: }
While MSS and wMFCC are sensitive to fine-grained spectral similarity and overall timbral evolution, respectively, they do not provide good intuition about the distance ``along'' the frequency axis.
For example, consider a ground truth signal containing a single frequency component at 100Hz, and two predictions, one at 500Hz and one at 5kHz.
Ignoring artifacts due to spectral leakage, these predictions would yield equal MSS distances from the ground truth, as the L1 distance is effectively a pointwise comparison.
While previous works on musical audio synthesis~\cite{hantrakul_fast_2019,engel_ddsp_2020} have employed fundamental frequency estimators such as CREPE~\cite{kim_crepe_2018} to compute a pitch similarity metric, we believe this would be unsuitable for the audio in our dataset, which is frequently aharmonic, aperiodic, or exhibits characteristics well outside the training distribution of common pitch estimators.

We thus adopt the \textit{spectral optimal transport} (SOT) distance, first proposed as a differentiable loss function for optimizing frequency parameters~\cite{torres_unsupervised_2024}, as a more permissive metric for similarity along the frequency axis.
The SOT distance follows the intuition that a greater difference between two frequencies should correspond to a larger distance.
It is computed as the Wasserstein-1 distance between normalized magnitude spectra:

\begin{equation}
  \mathcal{W}_1(\mu, \nu) = \inf_{\pi\in\Gamma(\mu, \nu)} \int |x - y| d\pi(x, y),
\end{equation}

\noindent where $\mu$ and $\nu$ are probability measures, here defined as a sum of Dirac measures located at Fourier harmonic frequencies and weighted by their magnitudes in the ground truth and predicted spectra respectively, and $\Gamma$ is the set of all couplings between elements in $\mu$ and $\nu$.
The infimum thus corresponds to the optimal transport cost between the two measures. In this one-dimensional case, a simple closed form solution is available.

\textbf{RMS: }
Finally, to compare the overall amplitude evolution of the predicted signal to the ground truth, we compute the framewise root mean squared energy envelopes of the two signals and take their cosine similarity.
In this way, we are able to compare the manner in which signal amplitude evolves over time while remaining invariant to overall amplitude scaling errors.

\section{Learnt \ptot{} parameters}

We provide plots of the assignment matrix $\A$ and parameter in/out embeddings $\Z$ and $\Z^\prime$ learned by our $\ptot{}$ module on the $k$-osc and Surge XT tasks.
The assignment matrices have been sorted lexicographically by row (where each row corresponds to a token), to allow easier identification of patterns.
Axes are labelled by groups of parameters, where possible. For detailed information about which parameters are present, see Table~\ref{tab:surge}.

Figs~\ref{fig:ass_kosc} and~\ref{fig:ass_kosc_asym} display the assignment matrices learned by the \ptot{} module on the symmetric and asymmetric variants, respectively, of the $k$-osc task for $k=8$.
We see clearly that, on the symmetric task, the matrix organizes the parameters such that each oscillator is represented by a single token. 
This is equivalent to the tokenization scheme of the \cnfe{} variant.
Conversely, the asymmetric task leads to a disorganized assignment matrix where there is no clear correspondence between oscillators and tokens.

Figs~\ref{fig:embed_kosc} and~\ref{fig:embed_kosc_asym}, again symmetric/asymmetric respectively, illustrate the learned parameter embeddings.
In particular, these are represented as self-similarity matrices.
Again, the symmetric task leads to clear structure. In fact, it appears the same vector (up to scaling) has been learned for parameters within the same group (e.g. all frequencies are represented by the same vector regardless of oscillator).
The asymmetric variant, once more, leads to no clear organization.

Figs~\ref{fig:ass_simp} and~\ref{fig:embed_simple} illustrate the assignment and embedding matrices for the Simple variant of the Surge XT task.
While interpreting such an assignment matrix becomes challenging beyond a handful of parameters, we note that the embeddings show clearly that the model has learned to represent the internal structure of the synthesizer.
In particular, we see repeated patterns indicating self-similarity for symmetric components such as oscillators, filters, and LFOs.
We observe a similar phenomenon, at a larger scale, in Figs~\ref{fig:ass_full} and~\ref{fig:embed_full}, which illustrate the learned parameters on the Surge Full dataset.

\begin{figure*}[t]
  \includegraphics[width=\textwidth]{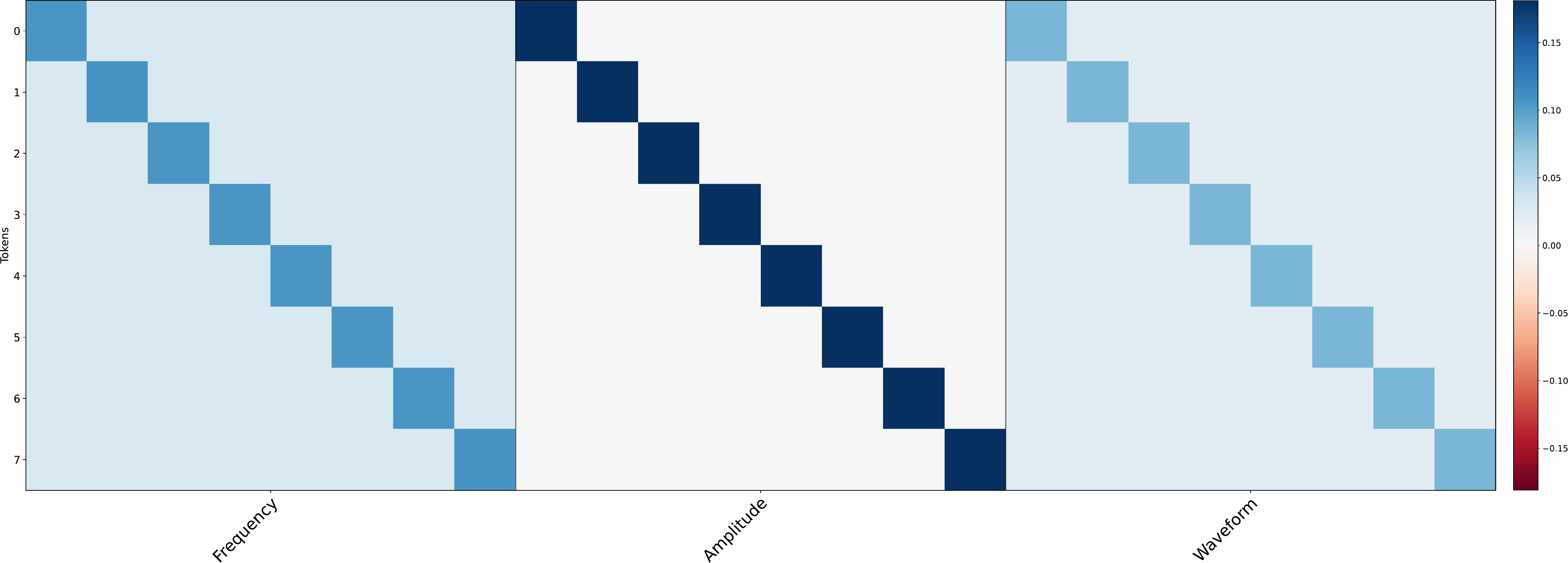}
  \caption{Learned \ptot{} assignment matrix $\A$ on \textbf{symmetric} $k$-osc task for $k=8$. Rows are lexicographically sorted by their indices sorted by descending bin value. The final matrix appropriately groups the parameters of each oscillator into an independent token.}\label{fig:ass_kosc}
\end{figure*}
\begin{figure*}[t]
  \includegraphics[width=\textwidth]{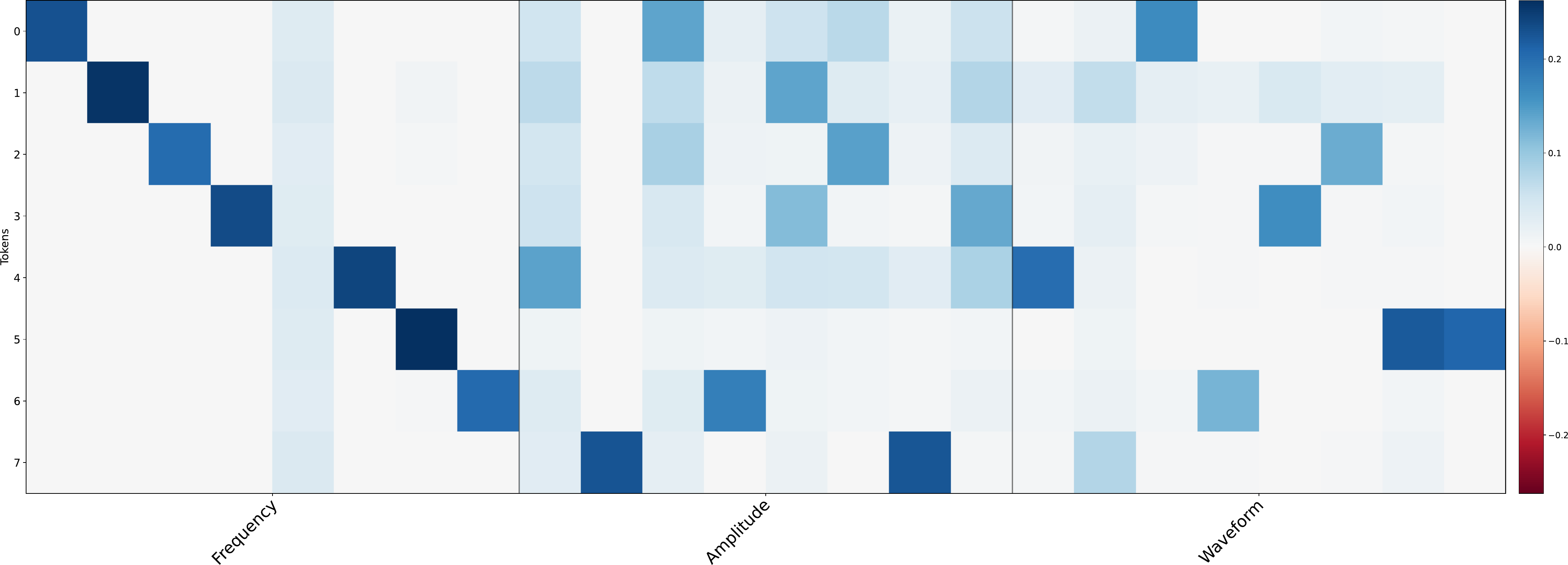}
  \caption{Learned \ptot{} assignment matrix $\A$ on \textbf{asymmetric} $k$-osc task for $k=8$. Rows are lexicographically sorted by their indices sorted by descending bin value.}\label{fig:ass_kosc_asym}
\end{figure*}
\begin{figure*}[t]
  \includegraphics[width=\textwidth]{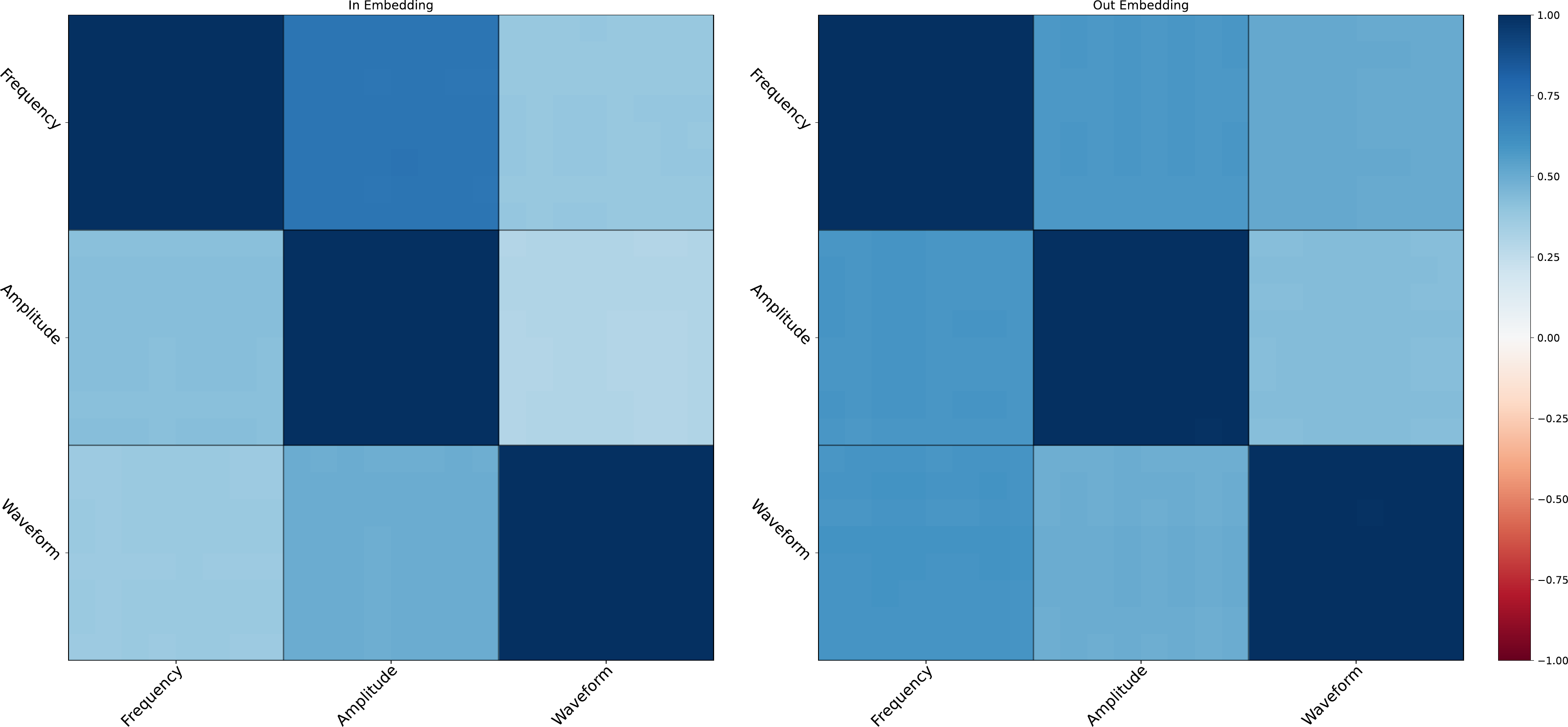}
  \caption{Cosine self-similarity of learned \ptot{} parameter embeddings $\Z$ and $\Z^\prime$ on \textbf{symmetric} $k$-osc task for $k=8$. Note that the blocks of solid colour are in fact $8\times 8$ squares of pixels --- the model has simply learned the same embedding for each parameter group.}\label{fig:embed_kosc}
\end{figure*}
\begin{figure*}[t]
  \includegraphics[width=\textwidth]{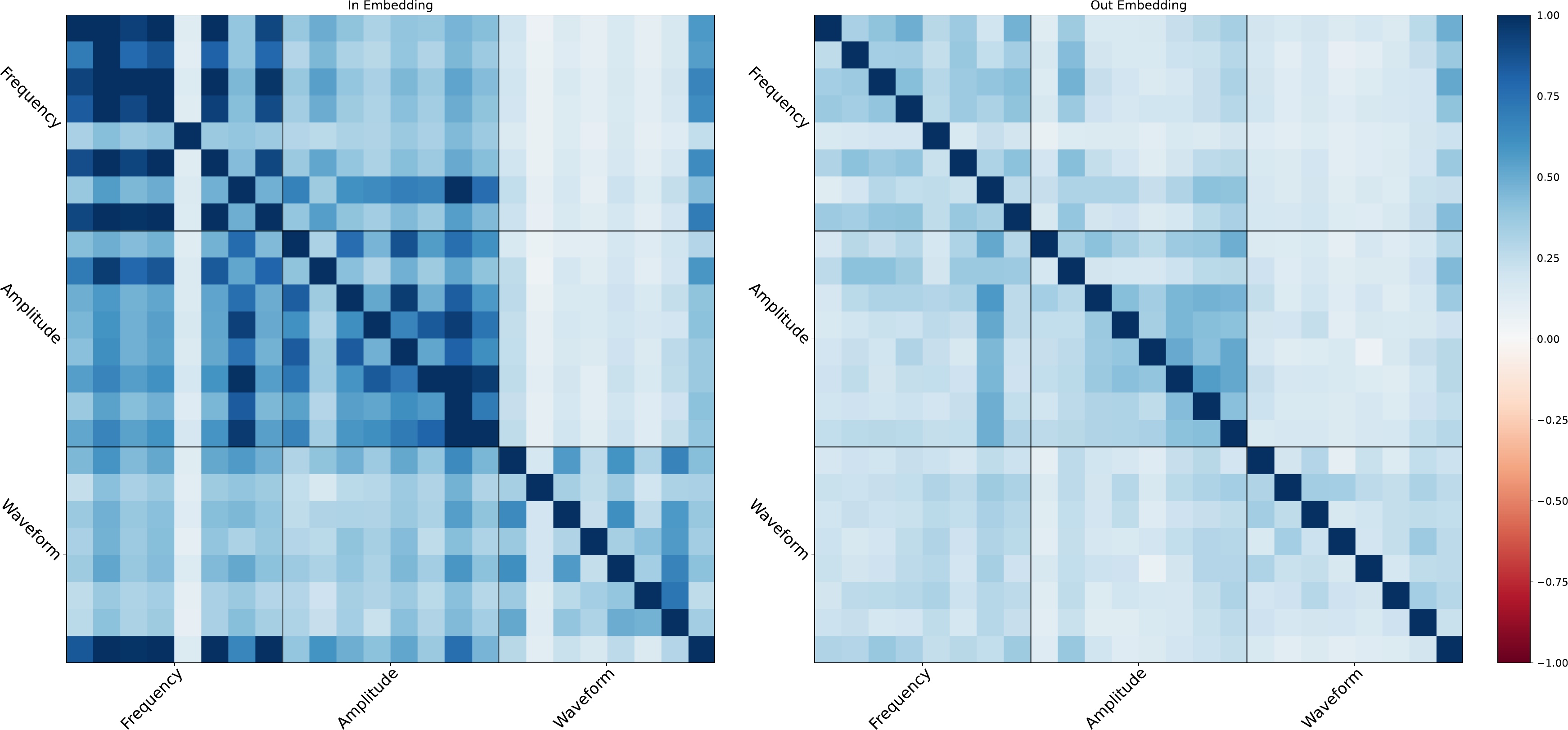}
  \caption{Cosine self-similarity of learned \ptot{} parameter embeddings $\Z$ and $\Z^\prime$ on \textbf{asymmetric} $k$-osc task for $k=8$. The model has not learned any meaningful symmetric structure.}\label{fig:embed_kosc_asym}
\end{figure*}

\begin{figure*}[t]
  \includegraphics[width=\textwidth]{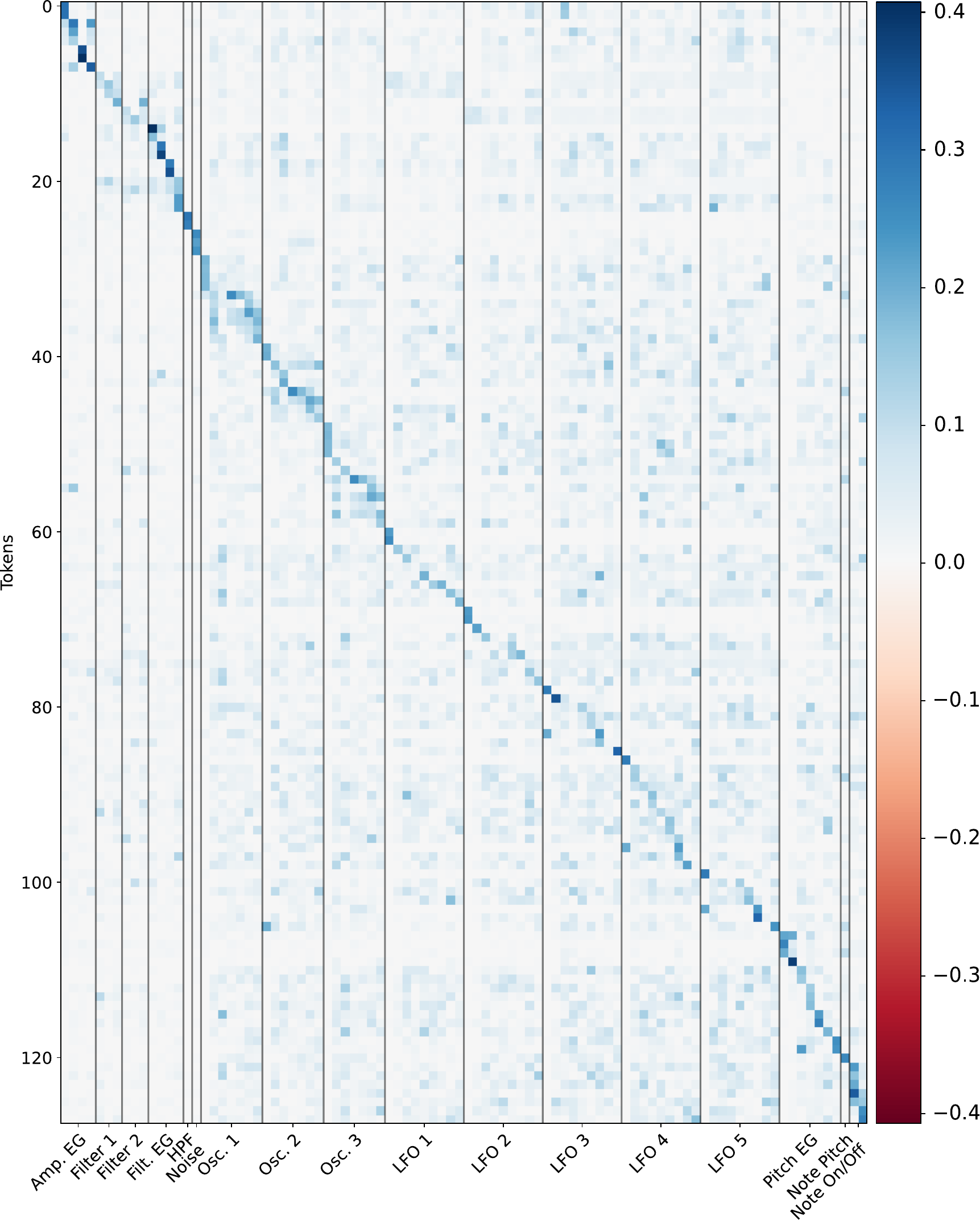}
  \caption{Learned \ptot{} assignment matrix $\A$ on \textit{Surge Simple} task. Rows are lexicographically sorted by their indices sorted by descending bin value.}\label{fig:ass_simp}
\end{figure*}
\begin{figure*}[t]
  \includegraphics[width=\textwidth]{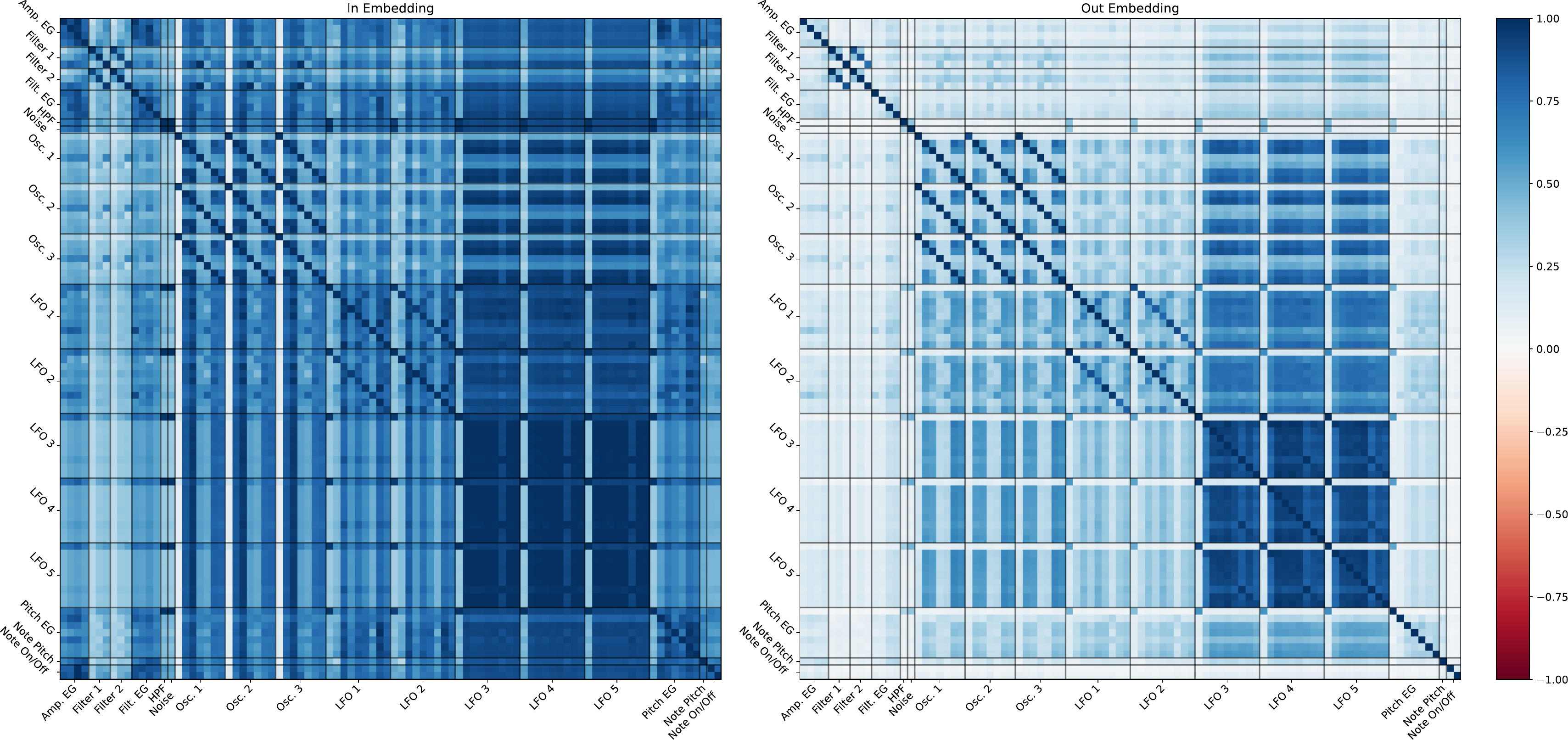}
  \caption{Cosine self-similarity of learned \ptot{} parameter embeddings $\Z$ and $\Z^\prime$ on \textit{Surge Simple} task. Note the repeated structures corresponding to symmetric subsystems of the synthesizer.}\label{fig:embed_simple}
\end{figure*}

\begin{figure*}[t]
  \includegraphics[width=\textwidth]{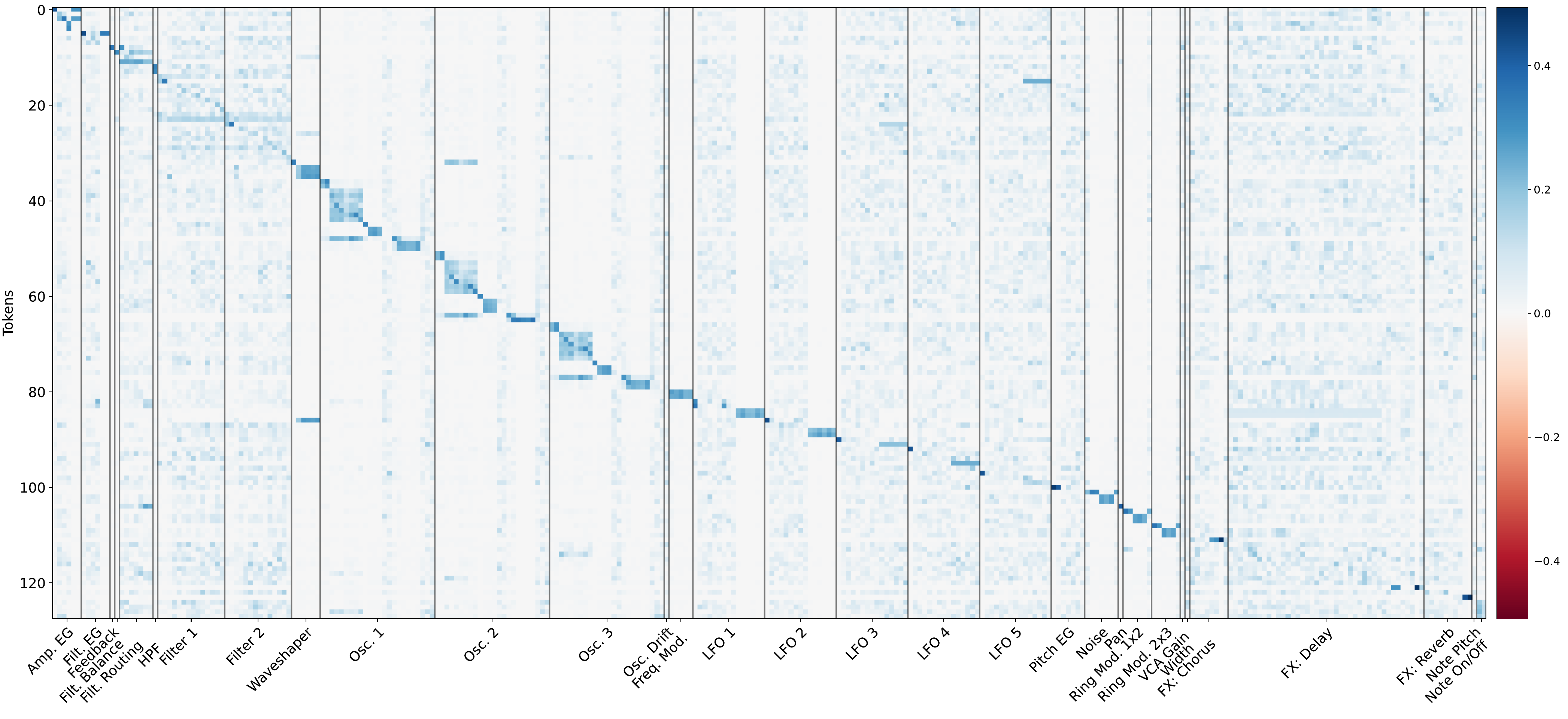}
  \caption{Learned \ptot{} assignment matrix $\A$ on \textit{Surge Full} task. Rows are lexicographically sorted by their indices sorted by descending bin value.}\label{fig:ass_full}
\end{figure*}
\begin{figure*}[t]
  \includegraphics[width=\textwidth]{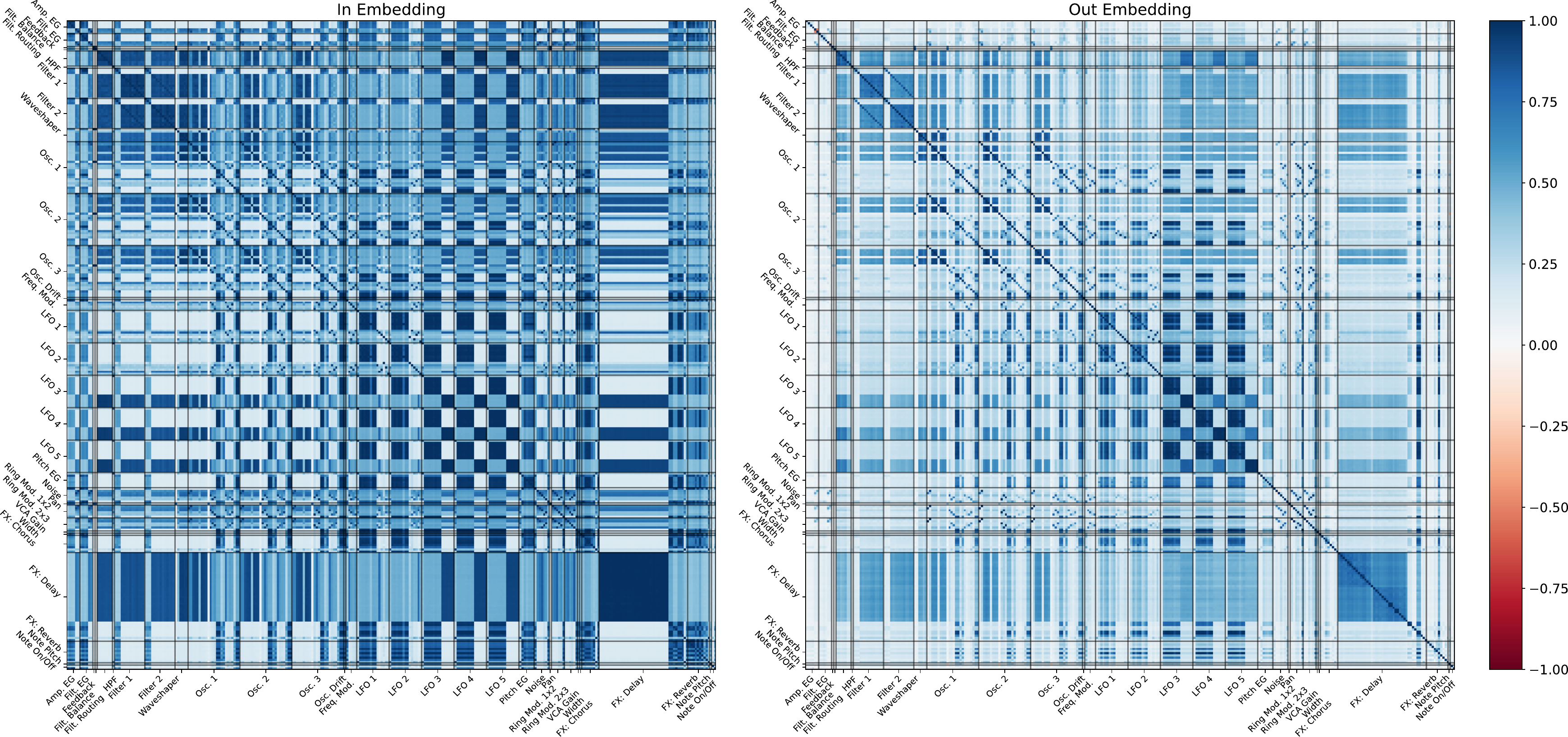}
  \caption{Cosine self-similarity of learned \ptot{} parameter embeddings $\Z$ and $\Z^\prime$ on \textit{Surge Full} task. Note the repeated structures corresponding to symmetric subsystems of the synthesizer.}\label{fig:embed_full}
\end{figure*}



\section{Surge XT dataset}

Surge XT is described as a ``hybrid'' synthesizer, which incorporates multiple different methods for producing sound, many of which are capable of introducing non-trivial permutation symmetries.

The synthesizer consists of multiple ``scenes'', each of which contains a copy of the synthesizer's sound generation engine. 
The outputs of each scene can be routed to a mixer for aggregation and further processing, while MIDI inputs to each scene's sound engine can be determined by different voice split and layering rules.
For the present experiment, we disable the second scene and focus on just one instance of the sound engine.

Due to the complexity of the audio effects module in Surge XT, and the wide variety of effect types,\footnote{Surge XT even contains a neural network as an audio effect!} we choose to adopt a fixed effect routing for our experiments, using only the Chorus, Delay, and Reverb II effects.

Similarly, due to the flexibility of Surge XT's modulation routing and large number of multi-functional modulation sources, we adopt a fixed routing which reflects common usage of LFOs and auxiliary envelopes.
Specifically, we map LFOs 1 and 2 to the cutoff frequencies of the two filters, and LFOs 3-5 to the pulse width parameters of the three oscillators.
We expose a limited subset of the mode parameters of each of these modulation sources, corresponding to a selection of different LFO and noise sources, and excluding envelope generators, sequencers, and programmable function generators.
We fix LFO 6 to the attack-decay-sustain-release (ADSR) envelope generator mode and map it to the global pitch modulation. 
We thus refer to it as "Pitch EG", and its amplitude as "Pitch EG Mod Amount". 

In general, we endeavoured to retain as many synthesis parameters as possible in the prediction task, deactivating only those which introduced an unnecessary redundancy (e.g. global volume and scene volume are equivalent in our configuration) or which are irrelevant to our particular experimental setup (e.g. filter keytracking, LFO trigger mode, pitch bend amount, etc.)
For each of the Simple and Full variants of the dataset, we construct a base preset which sets unused parameters to sensible defaults --- for example, we mute the contribution of scene B.
These presets are provided in the source code repository.
Details of the parameters used in both the Simple and Full versions of the Surge XT task are provided in Table~\ref{tab:surge}.

All parameters, including discrete values, are passed to the VST as floating point values on the interval $[0, 1]$.
These values are then internally scaled and quantized by the plugin.
For training, we rescale parameters to the interval $[-1, 1]$ in accordance with common practice for flow-based and diffusion-based models.
We note that parameter scaling is known to influence performance on sound matching~\cite{han_learning_2024} and will explore this in more detail in a subsequent publication.

For most continuous parameters we simply uniformly sample from the unit interval.
For most discrete parameters, we sample uniformly from the possible values, and use a one-hot representation.
Thus, each discrete parameter occupies as many spaces in the parameter vector as it has possible values.
In certain cases we increase the probability of selecting a particular value, where it is reasonable to expect that this value is the usual or ``default'' setting for this parameter. For example, each LFO's amplitude has 50\% probability of being 0, and will otherwise be sampled from the unit interval. Similarly, the waveshaper type is set to ``None'' with probability 50\%, and otherwise sampled uniformly from the list from remaining types.
We also truncate the sampling interval for certain parameter ranges where extreme values could cause issues with clipping (e.g. the upper bound for certain gain parameters is as high as +48dB), or where certain values do not make sense in our experimental setup (e.g. an amplitude envelope with an attack time of length 30s is unlikely to produce a usable sound given the length of our MIDI note event).
Rather than listing all sampling distributions here, we instead refer the reader to our source code repository for full details: \url{https://github.com/ben-hayes/synth-permutations} 

While previous sound matching work typically uses fixed MIDI note pitches and/or durations, here we opt to view these as predictable parameters in order to increase the adaptability of our system to a variety of different inputs.
We thus uniformly sample MIDI pitches from a two octave interval centered at C4 (MIDI note 60) and uniformly sample note-on and note-off event times within the 4 second window.


\onecolumn
\afterpage{%
  \begin{longtable}{rlccp{8.5cm}}
 \addlinespace[-\aboverulesep] 
 \cmidrule[\heavyrulewidth]{2-5}
    &
     Parameter &
     Simple &
     Full &
     Notes \\\cmidrule{2-5}
\ldelim\{{4}{*}[${}2\times$]\hspace{-1em} &
Filter $i$ Cutoff &
\cmark & \cmark &
 \\
 &
Filter $i$ FEG Mod. Amount &
\cmark & \cmark &
 \\
 &
Filter $i$ Resonance &
\cmark & \cmark &
 \\
 &
Filter $i$ Type &
\xmark & \cmark &
Options include low-, band-, and high-pass filters at 12dB and 24dB per octave, as well as notch, comb, and allpass filters. \\
\ldelim\{{10}{*}[${}5\times$]\hspace{-1em} &
LFO $j$ Amplitude &
\cmark & \cmark &
 \\
 &
LFO $j$ Attack &
\cmark & \cmark &
 \\
 &
LFO $j$ Decay &
\cmark & \cmark &
 \\
 &
LFO $j$ Deform &
\cmark & \cmark &
Continuously warps the shape of the LFO. \\
 &
LFO $j$ Hold &
\cmark & \cmark &
 \\
 &
LFO $j$ Phase &
\cmark & \cmark &
 \\
 &
LFO $j$ Rate &
\cmark & \cmark &
 \\
 &
LFO $j$ Release &
\cmark & \cmark &
 \\
 &
LFO $j$ Sustain &
\cmark & \cmark &
 \\
 &
LFO $j$ Type &
\xmark & \cmark &
Options include multiple waveform shapes, random sources, and envelope generators. \\
\ldelim\{{12}{*}[${}3\times$]\hspace{-1em} &
Osc. $k$ Mute &
\xmark & \cmark &
 \\
 &
Osc. $k$ Octave &
\xmark & \cmark &
 \\
 &
Osc. $k$ Pitch &
\cmark & \cmark &
 \\
 &
Osc. $k$ Pulse &
\cmark & \cmark &
 \\
 &
Osc. $k$ Route &
\xmark & \cmark &
Oscillator signal patch point. \\
 &
Osc. $k$ Sawtooth &
\cmark & \cmark &
 \\
 &
Osc. $k$ Sync &
\cmark & \cmark &
 \\
 &
Osc. $k$ Triangle &
\cmark & \cmark &
 \\
 &
Osc. $k$ Unison Detune &
\xmark & \cmark &
 \\
 &
Osc. $k$ Unison Voices &
\xmark & \cmark &
 \\
 &
Osc. $k$ Volume &
\cmark & \cmark &
 \\
 &
Osc. $k$ Width &
\cmark & \cmark &
 \\
 &
Amp. EG Attack &
\cmark & \cmark &
 \\
 &
Amp. EG Decay &
\cmark & \cmark &
 \\
 &
Amp. EG Env. Mode &
\xmark & \cmark &
Options are Digital and Analogue. These change the curvature of envelope segments. \\
 &
Amp. EG Release &
\cmark & \cmark &
 \\
 &
Amp. EG Sustain &
\cmark & \cmark &
 \\
 &
Feedback &
\xmark & \cmark &
Activated conditionally depending on Filter Routing. \\
 &
Filter Balance &
\xmark & \cmark &
Affects the mix between the two filters. \\
 &
Filter Routing &
\xmark & \cmark &
Selects from a number of internal routing templates that include various parallel/cascade configurations, feedback paths, and waveshaper positions. \\
 &
Filter EG Attack &
\cmark & \cmark &
 \\
 &
Filter EG Decay &
\cmark & \cmark &
 \\
 &
Filter EG Env. Mode &
\xmark & \cmark &
Same options as Amp. EG Env. Mode. \\
 &
Filter EG Release &
\cmark & \cmark &
 \\
 &
Filter EG Sustain &
\cmark & \cmark &
 \\
 &
FM Depth &
\xmark & \cmark &
Behaviour is conditional on FM Routing. \\
 &
FM Routing &
\xmark & \cmark &
Allows frequency modulation between the oscillators in different configurations. \\
 &
Highpass &
\cmark & \cmark &
As well as the two configurable filters, there is a global highpass filter. \\
&
Pitch EG Mod. Amount &
\cmark & \cmark & Pitch EG implemented via LFO 6 set to ADSR mode.
 \\
 &
Pitch EG Attack &
\cmark & \cmark &
 \\
 &
Pitch EG Decay &
\cmark & \cmark &
 \\
 &
Pitch EG Deform &
\cmark & \cmark &
 \\
 &
Pitch EG Hold &
\cmark & \cmark &
 \\
 &
Pitch EG Release &
\cmark & \cmark &
 \\
 &
Pitch EG Sustain &
\cmark & \cmark &
 \\
 &
Noise Color &
\xmark & \cmark &
 \\
 &
Noise Mute &
\xmark & \cmark &
 \\
 &
Noise Route &
\xmark & \cmark &
Noise signal patch point. \\
 &
Noise Volume &
\cmark & \cmark &
 \\
 &
Osc. Drift &
\xmark & \cmark &
Independently varies the pitch of all oscillators and their unison voices. \\
 &
Pan &
\xmark & \cmark &
 \\
 &
Ring Modulation 1x2 Mute &
\xmark & \cmark &
Ring modulation between oscillators is always running with fixed modulation depth, and can be mixed into the signal. \\
 &
Ring Modulation 1x2 Route &
\xmark & \cmark &
Ring Mod. signal patch point. \\
 &
Ring Modulation 1x2 Volume &
\xmark & \cmark &
 \\
 &
Ring Modulation 2x3 Mute &
\xmark & \cmark &
 \\
 &
Ring Modulation 2x3 Route &
\xmark & \cmark &
Ring Mod. signal patch point.  \\
 &
Ring Modulation 2x3 Volume &
\xmark & \cmark &
 \\
 &
VCA Gain &
\xmark & \cmark &
 \\
 &
Waveshaper Drive &
\xmark & \cmark &
 \\
 &
Waveshaper Type &
\xmark & \cmark &
 Choice between soft, hard, asymmetric, and sine waveshapers. \\
 &
Width &
\xmark & \cmark &
 \\
 &
Chorus Delay Feedback &
\xmark & \cmark &
 \\
 &
Chorus Delay Time &
\xmark & \cmark &
 \\
 &
Chorus EQ High Cut &
\xmark & \cmark &
 \\
 &
Chorus EQ Low Cut &
\xmark & \cmark &
 \\
 &
Chorus Modulation Depth &
\xmark & \cmark &
 \\
 &
Chorus Modulation Rate &
\xmark & \cmark &
 \\
 &
Chorus Output Mix &
\xmark & \cmark &
 \\
 &
Chorus Output Width &
\xmark & \cmark &
 \\
 &
Delay Time Left &
\xmark & \cmark &
 \\
 &
Delay Time Right &
\xmark & \cmark &
 \\
 &
Delay Feedback EQ Crossfeed &
\xmark & \cmark &
 \\
 &
Delay Feedback EQ Feedback &
\xmark & \cmark &
 \\
 &
Delay Feedback EQ High Cut &
\xmark & \cmark &
 \\
 &
Delay Feedback EQ Low Cut &
\xmark & \cmark &
 \\
 &
Delay Input Channel &
\xmark & \cmark &
 \\
 &
Delay Modulation Depth &
\xmark & \cmark &
 \\
 &
Delay Modulation Rate &
\xmark & \cmark &
 \\
 &
Delay Output Mix &
\xmark & \cmark &
 \\
 &
Delay Output Width &
\xmark & \cmark &
 \\
 &
Reverb EQ HF Damping &
\xmark & \cmark &
 \\
 &
Reverb EQ LF Damping &
\xmark & \cmark &
 \\
 &
Reverb Output Mix &
\xmark & \cmark &
 \\
 &
Reverb Output Width &
\xmark & \cmark &
 \\
 &
Reverb Pre-delay &
\xmark & \cmark &
 \\
 &
Reverb Buildup &
\xmark & \cmark &
 \\
 &
Reverb Decay Time &
\xmark & \cmark &
 \\
 &
Reverb Diffusion &
\xmark & \cmark &
 \\
 &
Reverb Modulation &
\xmark & \cmark &
 \\
 &
Reverb Room Size &
\xmark & \cmark &
 \\
&
Note Start &
\cmark & \cmark &
MIDI note parameter. \\
 &
Note End &
\cmark & \cmark &
MIDI note parameter. \\
 &
Note Pitch &
\cmark & \cmark &
MIDI note parameter. \\

 \cmidrule[\heavyrulewidth]{2-5}
 \addlinespace[-\belowrulesep] 
\rule{0pt}{0ex}\\
\caption{Parameters of the Surge XT synthesizer used in the \textit{Surge Simple} and \textit{Surge Full} datasets.\label{tab:surge}}
\end{longtable}
}

\clearpage
\twocolumn
\end{document}